\pdfoutput=1 
\documentclass[a4paper,11pt,fleqn]{article}
\usepackage[usenames,dvipsnames]{pstricks}
\usepackage{comment}
\usepackage[official]{eurosym}
\usepackage{euscript}
\usepackage{charter}

\usepackage{multirow}
\usepackage{array}
\usepackage{booktabs} 
\usepackage{multirow}
\usepackage{caption}
\usepackage{subcaption}

\usepackage{amsfonts}
\usepackage{amsmath}
\usepackage{amssymb}
\usepackage{theorem}

\usepackage{url} 

\definecolor{darkgreen}{cmyk}{0.8,0,0.8,0.45}
\definecolor{lightgreen}{cmyk}{0.8,0,0.8,0.25}
\usepackage[ruled,vlined,linesnumbered]{algorithm2e}

\SetCommentSty{mycommfont}
\SetKwInput{Parameter}{Parameters} 
\let\oldnl\nl
\newcommand{\nonl}{\renewcommand{\nl}{\let\nl\oldnl}}

\usepackage{graphicx}
\usepackage{epstopdf}
\DeclareGraphicsExtensions{.pdf,.jpg,.png,.eps}
\graphicspath{{figures/illustrations},{figures/results}}

\definecolor{labelkey}{rgb}{0,0.08,0.45}
\definecolor{refkey}{rgb}{0,0.6,0.0}
\definecolor{Brown}{rgb}{0.45,0.0,0.05}
\definecolor{dgreen}{rgb}{0.00,0.49,0.00}
\definecolor{dblue}{rgb}{0,0.08,0.75}
\RequirePackage[colorlinks,hyperindex]{hyperref}
\hypersetup{linktocpage=true,citecolor=dblue,linkcolor=dgreen}

\usepackage{cleveref}

\newtheorem{theorem}{Theorem}[section]

\newtheorem{proposition}[theorem]{Proposition}
\theoremstyle{plain}{\theorembodyfont{\rmfamily}%
}
\theoremstyle{plain}{\theorembodyfont{\rmfamily}%
\newtheorem{assumption}[theorem]{Assumption}}
\theoremstyle{plain}{\theorembodyfont{\rmfamily}%
\newtheorem{model}[theorem]{Model}}
\theoremstyle{plain}{\theorembodyfont{\rmfamily}%
}
\theoremstyle{plain}{\theorembodyfont{\rmfamily}%
\newtheorem{example}[theorem]{Example}}
\theoremstyle{plain}{\theorembodyfont{\rmfamily}%
\newtheorem{remark}[theorem]{Remark}}
\theoremstyle{plain}{\theorembodyfont{\rmfamily}%
}
\theoremstyle{plain}{\theorembodyfont{\rmfamily}%
}
\numberwithin{equation}{section}

\usepackage{enumitem}

\setlist[enumerate]{leftmargin=.5in}
\setlist[itemize]{leftmargin=.5in}

\renewcommand{\leq}{\ensuremath{\leqslant}}
\renewcommand{\geq}{\ensuremath{\geqslant}}
\renewcommand{\le}{\ensuremath{\leqslant}}

\newcommand{\minimize}[2]{\ensuremath{\underset{\substack{{#1}}}%
{\text{\rm minimize}}\;\;#2 }}

\newcommand{\menge}[2]{\big\{{#1}~\big |~{#2}\big\}} 
 
\newcommand{\HHH}{{\ensuremath{\boldsymbol{\mathsf H}}}}
\newcommand{\GGG}{{\ensuremath{\boldsymbol{\mathsf G}}}}

\newcommand{\HH}{\ensuremath{{\mathsf H}}}
\newcommand{\GG}{\ensuremath{{\mathsf G}}}

\newcommand{\emp}{\ensuremath{{\varnothing}}}

\newcommand{\Idb}{\ensuremath{\mathbf{I}}\,}

\newcommand{\cart}{\ensuremath{\raisebox{-0.5mm}{\mbox{\Large{$\times$}}}}}
\newcommand{\Cart}{\ensuremath{\raisebox{-0.5mm}{\mbox{\LARGE{$\times$}}}}}
\newcommand{\RR}{\ensuremath{\mathbb{R}}}
\newcommand{\RPP}{\ensuremath{\left]0,+\infty\right[}}
\newcommand{\RX}{\ensuremath{\left]-\infty,+\infty\right]}}
\newcommand{\EE}{\ensuremath{\mathbb E}}
\newcommand{\NN}{\ensuremath{\mathbb N}}

\newcommand{\prox}{\ensuremath{\text{\rm prox}}}

\newcommand{\argmind}[2]{\ensuremath{\underset{\substack{{#1}}}%
{\mathrm{argmin}}\;\;#2 }}

\newcommand{\Ab}{\ensuremath{\boldsymbol{A}}}

\newcommand{\Db}{\ensuremath{\boldsymbol{D}}}  \newcommand{\db}{\ensuremath{\boldsymbol{d}}}  
  \newcommand{\eb}{\ensuremath{\boldsymbol{e}}}

  \newcommand{\ub}{\ensuremath{\boldsymbol{u}}}  
  \newcommand{\vb}{\ensuremath{\boldsymbol{v}}}  
  \newcommand{\wb}{\ensuremath{\boldsymbol{w}}}  
  \newcommand{\xb}{\ensuremath{\boldsymbol{x}}}  
  \newcommand{\yb}{\ensuremath{\boldsymbol{y}}}  
  \newcommand{\zb}{\ensuremath{\boldsymbol{z}}}  

\newcommand{\alphab}{\ensuremath{\boldsymbol{\alpha}}}
\newcommand{\betab}{\ensuremath{\boldsymbol{\beta}}}
\newcommand{\deltab}{\ensuremath{\boldsymbol{\delta}}}
\newcommand{\xib}{\ensuremath{\boldsymbol{\xi}}}

\newcommand{\Dc}{\ensuremath{\mathcal{D}}}
\newcommand{\Hc}{\ensuremath{\mathcal{H}}}

\newcommand{\Nc}{\ensuremath{\mathcal{N}}}
\newcommand{\Pc}{\ensuremath{\mathcal{P}}}
\newcommand{\Rc}{\ensuremath{\mathcal{R}}}
\newcommand{\Sc}{\ensuremath{\mathcal{S}}}

\newcommand{\Wc}{\ensuremath{\mathcal{W}}}

\newcommand{\eE}{\ensuremath{\mathbb{E}}}

\newcommand{\eN}{\ensuremath{\mathbb{N}}}
\newcommand{\eR}{\ensuremath{\mathbb{R}}}
\newcommand{\eV}{\ensuremath{\mathbb{V}}}

\usepackage{amsopn}

\DeclareMathOperator{\SNR}{SNR}
\DeclareMathOperator{\SSIM}{SSIM}

\newcommand{\dist}[2]{\ensuremath{\pi_{#1}\left(#2\right)}}
\newcommand{\distc}[3]{\ensuremath{\pi_{#1}\left(#2 \mid #3\right)}}

\newcommand{\xmap}{\xb_{\text{MAP}}}
\newcommand{\xmmse}{\xb_{\text{MMSE}}}

\newcommand{\ncomposite}{C}
\newcommand{\idcomposite}{c}

\newcommand{\nworkers}{K}
\newcommand{\nmc}{N_\text{MC}}
\newcommand{\nbi}{N_\text{bi}}
\newcommand{\tflop}{\tau_{\texttt{flop}}}
\newcommand{\tlatency}{\tau_{\texttt{latency}}}
\newcommand{\tband}{\tau_{\texttt{bandwidth}}}

\newcommand{\email}{}
\newcommand{\production}{0}

\begin{document}

\title{\sffamily 
A Distributed Block-Split Gibbs Sampler with \\ Hypergraph Structure for High-Dimensional \\ Inverse Problems}
\author{
P.-A. Thouvenin$^{\ddagger}$, A. Repetti$^{\dagger\star}$, P. Chainais$^{\ddagger}$ \footnote{PC and PAT acknowledge support from the ANR project ``Chaire IA Sherlock'' ANR-20-CHIA-0031-01, the {\em programme d'investissements d'avenir} ANR-16-IDEX-0004 ULNE and Région HDF. AR acknowledges support from the Royal Society of Edinburgh.}
\\[5mm]
\small
\small $^\ddagger$ Universit{\'e} de Lille, CNRS, Centrale Lille, UMR 9189 CRIStAL, F-59000 Lille, France\\
\small $^\dagger$ School of Mathematics and Computer Sciences, Heriot-Watt University, Edinburgh, UK\\
\small $^\star$ School of Engineering and Physical Sciences, Heriot-Watt University, Edinburgh, UK\\
\small \email{\{pierre-antoine.thouvenin, pierre.chainais\}@centralelille.fr, a.repetti@hw.ac.uk}
}
\date{}

\maketitle
\thispagestyle{empty}

\vskip 8mm

\begin{abstract}
    Sampling-based algorithms are classical approaches to perform Bayesian inference in inverse problems. They provide estimators with the associated credibility intervals to quantify the uncertainty on the estimators. Although these methods hardly scale to high dimensional problems, they have recently been paired with optimization techniques, such as proximal and splitting approaches, to address this issue. Such approaches pave the way to distributed samplers, splitting computations to make inference more scalable and faster. We introduce a distributed Split Gibbs sampler (SGS) to efficiently solve such problems involving distributions with multiple smooth and non-smooth functions composed with linear operators. The proposed approach leverages a recent approximate augmentation technique reminiscent of primal-dual optimization methods. It is further combined with a block-coordinate approach to split the primal and dual variables into blocks, leading to a distributed block-coordinate SGS. The resulting algorithm exploits the hypergraph structure of the involved linear operators to efficiently distribute the variables over multiple workers under controlled communication costs. It accommodates several distributed architectures, such as the Single Program Multiple Data and client-server architectures. Experiments on a large image deblurring problem show the performance of the proposed approach to produce high quality estimates with credibility intervals in a small amount of time. Supplementary material to reproduce the experiments is available online.
\end{abstract}


\section{Introduction}
\label{Sec:introduction}
This work focuses on sampling from a generic distribution of the form
\begin{equation} \label{eq:dist-gen-lin}
    \dist{}{\xb} \propto \exp ( - h(\xb) -f(\xb) - g(\Db \xb) ),
\end{equation}
where $h \colon \HHH \to ]-\infty,+\infty]$ is a Lipschitz-differentiable function, $f \colon \HHH \to \RX$ and $g\colon \GGG \to ]-\infty,+\infty]$ are possibly non-smooth functions, and $\Db \colon \HHH \to \GGG$ is a linear operator.
Such a distribution typically arises as the posterior distribution involved in imaging inverse problems, from which Bayesian estimators need to be formed \cite{Pereyra2016_jstsp}.
In the remainder, we will consider that $\HHH = \RR^{\overline{N}}$ and $\GGG = \RR^{\overline{M}}$, where $\overline{N}$ and $\overline{M}$ are very large.
Sampling from~\eqref{eq:dist-gen-lin} is challenging due to (i) the presence of the composite function $g \circ \Db$ and (ii) the large dimension of $\HHH$ and $\GGG$. To address these issues, this paper proposes a distributed MCMC algorithm which (i) leverages the approximate augmentation AXDA~\cite{Vono_etal_2020} to decouple (\emph{split}) functions involved in~\eqref{eq:dist-gen-lin}, and (ii) exploits the structure of~\eqref{eq:dist-gen-lin} to design a distributed sampler reminiscent of block-coordinate approaches in optimization.
We briefly discuss recent splitting-based distributed samplers in~\Cref{Ssec:introduction:distribtued_sampling}, and highlight some of their limitations. Scalable splitting optimization approaches are reviewed in~\Cref{Ssec:introduction:optimization} to motivate this paper. The proposed approach, which takes further inspiration from the optimization literature, is outlined in~\Cref{Ssec:introduction:contributions}.

\subsection{Sampling methods: splitting and distributed techniques}
\label{Ssec:introduction:distribtued_sampling}

Markov chain Monte Carlo (MCMC) algorithms are generic approaches providing estimates with associated credibility intervals~\cite{Robert2010}. They aim to generate a Markov chain that yields samples from the target  distribution~\eqref{eq:dist-gen-lin} in the stationary regime. Nevertheless, they are often considered computationally too expensive to handle high dimensional problems, especially when composite functions are involved. This is often the case with inverse problems in image processing. 
Over the last decade, many authors have proposed more versatile and scalable optimization-inspired MCMC algorithms~\cite{Durmus2018, Pereyra2016,Salim_Richtarik_2020}. 
These approaches exploit quantities repeatedly used in optimization to efficiently explore high dimensional parameter spaces, most often gradients and proximal operators\footnote{The proximal operator of a proper, lower semi-continuous function $f \colon \eR^N \rightarrow ]-\infty,+\infty]$ is defined for any $\yb\in \eR^N$ by~\cite{Hiriart1993}: $\displaystyle \prox_f (\yb) = \argmind{\xb \in \eR^N}{\bigl\{ f(\xb) + \| \xb - \yb \|^2_2/2 \bigr\}}$.}.

A splitting approach based on an \emph{asymptotically exact data augmentation} (AXDA) has also recently been proposed by~\cite{Vono_etal_2019, Vono_etal_2020}. 
Inspired by splitting optimization approaches \cite{Komodakis2015}, 
AXDA introduces auxiliary variables to split composite distributions. The density~\eqref{eq:dist-gen-lin} is then approximated by
\begin{equation} \label{eq:dist-SPA}
    \dist{(\alpha, \beta)}{\xb,\zb,\ub}
    \propto \exp \big( - h(\xb) -f(\xb) -g(\zb) - \phi_\alpha(\Db\xb,\zb-\ub) - \psi_\beta(\ub) \big),
\end{equation}
where $\phi_{\alpha} \colon \GGG \times \GGG \to ]-\infty, +\infty]$, $\psi_{\beta} \colon \GGG \to ]-\infty, +\infty]$, 
$\alpha$ controls the discrepancy between $\Db\xb$ and $\zb - \ub$, and $\beta$ is an augmentation parameter. The variable $\zb$ is to be interpreted as a splitting variable, and $\ub$ is an additional augmentation parameter. The role of $\phi_{\alpha}$ is to strongly couple $\Db\xb$ and $\zb - \ub$, while $\psi_{\beta}$ keeps $\ub$ small enough. This latter parameter is not mandatory, but improves the mixing properties of the sampler by a further decoupling between $\Db\xb$ and $\zb$~\cite{Vono_etal_2019}.
For appropriate choices of $\phi_\alpha$ and $\psi_\beta$, the marginal distribution of $\xb$ with respect to~\eqref{eq:dist-SPA} converges to the target distribution \eqref{eq:dist-gen-lin} as $(\alpha,\beta) \to (0,0)$. A Gibbs sampler is proposed in~\cite{Vono_etal_2019, Vono_etal_2020} to draw samples from~\eqref{eq:dist-SPA}, referred to as the split Gibbs sampler (SGS). The additional cost of adding new variables is compensated by the benefit of this divide-to-conquer strategy. Further details on this splitting technique are provided in Section~\ref{Ssec:SPA-simple}. 
Designing efficient algorithms to handle distributions of the form~\eqref{eq:dist-gen-lin} becomes even more challenging when potentially many composite functions are considered.
An extension of SGS has been considered~\cite{Vono_etal_2019,Rendell2021} for distributions involving $\ncomposite \in \mathbb{N}^*$ composite terms, with a density of the form
\begin{equation} \label{eq:dist-gen-split}
    \pi(\xb) \propto \exp \big( -h(\xb) -f(\xb) - \sum_{\idcomposite=1}^\ncomposite g_\idcomposite(\Db_\idcomposite \xb) \big),
\end{equation}
where for every $\idcomposite \in \{1, \ldots, \ncomposite\}$, $\GGG_\idcomposite = \RR^{\overline{M}_\idcomposite}$, $\overline{M}= \sum_{\idcomposite=1}^\ncomposite \overline{M}_\idcomposite$, $\Db_\idcomposite \colon \HHH \to \GGG_\idcomposite$ and $g_\idcomposite \colon \GGG_\idcomposite \to \RX$.
Applying AXDA to~\eqref{eq:dist-gen-split} leads to an approximation with density
\begin{multline}\label{eq:dist-gen-split-SPA}
     \pi_{(\alphab, \betab)} \big( \xb,(\zb_\idcomposite, \ub_\idcomposite)_{1 \le \idcomposite \le \ncomposite} \big)\\
     \propto \exp \Big( - h(\xb) - f(\xb) 
     -\sum_{\idcomposite=1}^\ncomposite \Big( g_\idcomposite(\zb_\idcomposite) + \phi_{\idcomposite,\alpha_\idcomposite}(\Db_\idcomposite \xb, \zb_\idcomposite-\ub_\idcomposite) + \psi_{\idcomposite, \beta_\idcomposite}(\ub_\idcomposite ) \Big) \Big),
\end{multline}
where $(\zb_\idcomposite,\ub_\idcomposite)_{1 \le \idcomposite \le \ncomposite}$ are auxiliary variables and, for $\idcomposite \in \{1, \ldots, \ncomposite\}$, $\phi_{\idcomposite, \alpha_\idcomposite} \colon \GGG_\idcomposite \times \GGG_\idcomposite \to \RX$ and $\psi_{\idcomposite, \beta_\idcomposite} \colon \GGG_\idcomposite \to \RX$. The variables $(\zb_\idcomposite,\ub_\idcomposite)_{1 \le \idcomposite \le \ncomposite}$ are conditionally independent, paving the way to a distributed implementation on a client-server architecture.

Only a few distributed samplers have been proposed in the literature~\cite{Rendell2021, Vono_etal_2019, Vono_etal_2020}.
However, these samplers only focus on distributing the splitting variables associated with the composite functions $g_\idcomposite \circ \Db_\idcomposite$ from \eqref{eq:dist-gen-split}, without decomposing the global variable $\xb$ into blocks. 
In particular, \cite{Rendell2021} propose a consensus-based approximate posterior distribution, addressed with a distributed Metropolis-within-Gibbs sampler relying on a \emph{client-server} architecture. The sampler exploits the conditional independence between the splitting variables to parallelize computations. 
This is especially relevant for data-distributed applications, in which a shared parameter value needs to be inferred from a dataset distributed over multiple workers, e.g., for distributed logistic regression.
However, this setting has several drawbacks when turning to high dimensional problems. 
First, the consensus constraint necessitates to duplicate the variables of interest on the different workers.
Second, the client-server architecture may induce communication bottlenecks, as all (\emph{clients}) workers need to communicate with the server.
It also exhibits limitations in terms of distribution flexibility, since it separates composite functions only, without splitting the high dimensional global variable of interest $\xb$ into blocks.
In high dimensions, it may be of interest to split the variable $\xb$ itself, not only the data or the composite functions.

To summarize, splitting techniques have been introduced in sampling methods. They have permitted to handle multiple composite functions in parallel. Distributed versions have also been proposed in the literature, but restricted to client-server approaches. The client-server architecture can be critical in a high dimensional setting. In addition, none of these methods considered splitting the high-dimensional variable of interest $\xb$ into blocks. Block-coordinate approaches can be necessary in practice to handle high-dimensional problems, e.g., when $\xb$ is an image with more than $10^6$ unknown parameters. Since the distributed split Gibbs sampler proposed in this work are essentially inspired by optimization approaches, 
the next paragraph reviews methods from the optimization literature that combine all at once composite-function splitting, variable splitting (i.e., block-coordinate approaches) and distributed techniques.

\subsection{Splitting, distributed and block-coordinate methods in optimization}
\label{Ssec:introduction:optimization}

Optimization-based inference consists in estimating a mode of the distribution~\eqref{eq:dist-gen-lin} (e.g., the maximum \emph{a posteriori} (MAP) estimator), defined as a solution to
\begin{equation}	\label{eq:min-gen-lin}
	\minimize{\xb \in \HHH} h(\xb) +f(\xb) + g(\Db \xb).
\end{equation}
Optimization algorithms aim to build sequences that asymptotically converges to a solution to problem~\eqref{eq:min-gen-lin}.
Problem~\eqref{eq:min-gen-lin} can be efficiently solved with proximal primal-dual methods \cite{Komodakis2015, Condat2013, Vu2013}, that can be seen as the optimization counterpart of the splitting approaches discussed in Section~\ref{Ssec:introduction:distribtued_sampling}. 
Similarly to AXDA, primal-dual methods rely on an auxiliary variable $\zb \in \GGG$ associated with the  composite function $g\circ \Db$. The variable $\xb \in \HHH$ is referred to as the \emph{primal variable} while $\zb$ is called the \emph{dual variable}, since it is associated with the dual space $\GGG$ induced by the operator $\Db$.
It is worth noticing that the SGS approach developed by \cite{Vono_etal_2019} was directly inspired by such an optimization splitting technique.
Splitting proximal algorithms benefit from many acceleration techniques (e.g., inertia \cite{Beck2009, Nesterov1983}, preconditioning \cite{Chouzenoux2014, Combettes_Vu_2014}), they are versatile, and scalable. In particular, they are highly parallelizable, and can be efficiently distributed to split the computational cost per iteration, under well-established theoretical guarantees~\cite{Bauschke2017CAMOTHS, Combettes2011, Komodakis2015}.

Designing efficient algorithms to handle distributions of the form~\eqref{eq:dist-gen-split} becomes more challenging when multiple composite functions are involved, especially when the dimensions of $\xb$ and $\yb$ increase. In this context, the counterpart of problem~\eqref{eq:dist-gen-split} is given by
\begin{equation}    \label{eq:g-multi}
\minimize{\xb \in \HHH} h(\xb) +f(\xb) + \sum_{\idcomposite=1}^\ncomposite g_\idcomposite(\Db_\idcomposite \xb).
\end{equation}
Primal-dual proximal algorithms~\cite{Condat2013, Vu2013, Komodakis2015, Pesquet2015, Combettes2018} permit to handle composite terms in~\eqref{eq:g-multi} in parallel in the dual domains induced by operators $(\Db_\idcomposite)_{1 \le \idcomposite \le \ncomposite}$. This is possible since dual variables can be distributed on multiple workers. 

\begin{figure}[!t]
    \centering
    \includegraphics[width=7cm]{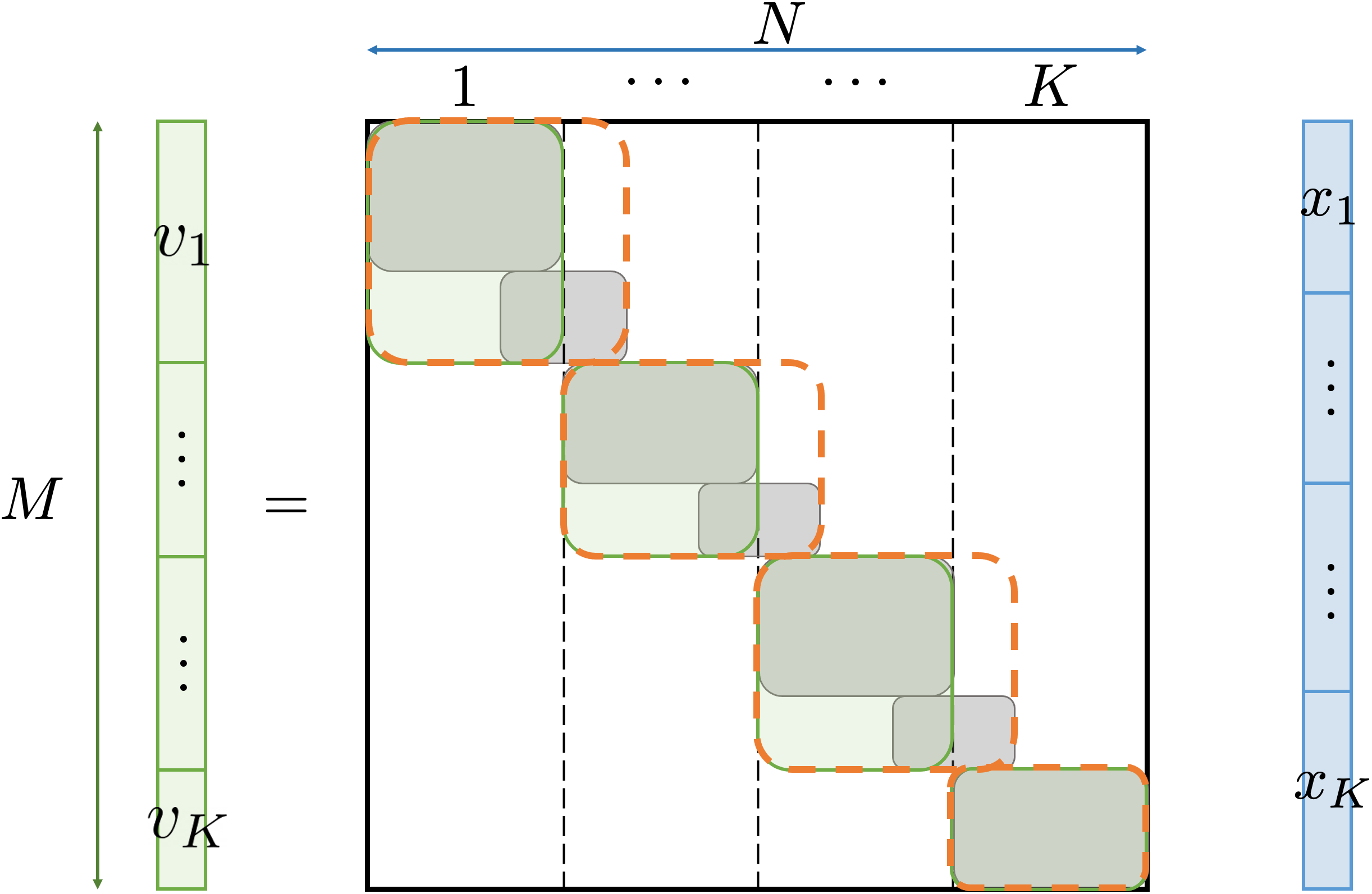}
    \vspace*{-0.3cm}
    \caption{\label{fig:blocksparse}\small
        Example of the block-sparse structure of the matrix $\Db$ involved in~\eqref{eq:dist-gen-lin}. 
        The columns of $\Db$ are split into $K$ contiguous blocks (black dashed lines), with a small overlap compared to the size of the blocks. 
        The orange dashed rectangles highlight the subparts of $\Db$ implemented on each worker $k$. These act on a parameter block stored on worker $k$, and a few parameters stored on worker $k+1$.
    }
\end{figure}

Further, to address problems with high dimensional variable $\xb$, a usual strategy in optimization consists in splitting $\xb$ into blocks $(\xb_k)_{1 \le k \le K}$, and either alternate between the blocks \cite{Bolte2014, Chouzenoux2016, Tseng92, TsengBCD2001}, or distribute the blocks over multiple workers \cite{Combettes2015, Pesquet2015}. 
In particular, \cite{Pesquet2015} combine such block-coordinate approaches with primal-dual splitting techniques to parallelize and distribute both the primal and dual variables. 
Then, the associated minimization problem is of the form
\begin{equation}	\label{eq:min-gen-multi}
	\minimize{\xb =(\xb_k)_{1\le k \le K} \in \HHH} \sum_{k=1}^K h_k(\xb_k) +f_k(\xb_k) + \sum_{\idcomposite=1}^\ncomposite g_\idcomposite(\sum_{k=1}^K \Db_{\idcomposite,k} \xb_k),
\end{equation}
where, for every $k\in \{1, \ldots, K\}$, $h_k$ and $f_k$ only act on the $k$-th block of the variable $\xb \in \HHH$.
Such a formulation is of particular interest when considering block-sparse matrices $\Db_\idcomposite$, as illustrated in Figure~\ref{fig:blocksparse}.
Finally, \cite{Pesquet2015} also developed asynchronous distributed algorithms over hypergraphs\footnote{Hypergraphs generalize structures of graphs, where edges can connect multiple nodes (i.e., variables) together, hence generalizing communications between variables.} structures by combining the resulting block-coordinate primal-dual algorithms with consensus constraints that impose all $(\xb_k)_{1 \le k \le K}$ to be equal. This strategy provides a high flexibility in the choice of the distribution architecture. Section~\ref{Ssec:dist:hypergraph} will detail the interest of hypergraphs.
Unfortunately, these algorithms only provide a point estimate, without additional information. 
In absence of ground truth, these approaches do not directly quantify the uncertainty over the estimate. 
The proposed approach will focus on drawing samples from the distribution corresponding to~\eqref{eq:min-gen-multi} using a distributed architecture.

\subsection{Proposed distributed block-coordinate SGS}
\label{Ssec:introduction:contributions}

This work focuses on cases where $\xb$ belongs to a high dimensional space, for instance for $\overline{N}$ larger than $10^6$. We introduce a distributed block-coordinate SGS to sample from~\eqref{eq:dist-gen-split} by splitting the global variable of interest $\xb$ into blocks $(\xb_k)_{1 \le k \le K}$.
This approach will also be able to handle multiple composite functions as in~\eqref{eq:min-gen-multi}. 
To this aim, we pair SGS with 
optimization-inspired MCMC transition kernels \cite{Durmus2018, Pereyra2016, Salim_Richtarik_2020} to sample from conditional distributions that would otherwise be either intractable or challenging to handle in a distributed setting.

Some steps of the MCMC sampler will correspond to term-wise operations, which are easy to distribute. Some others will involve the linear operator $\Db$ and its adjoint $\Db^*$, calling for communications. This work is aimed at designing an efficient distributed version of the corresponding serial sampler, in particular by dealing with terms that involve $\Db$ or $\Db^*$.

The distributed structure associates each block $\xb_k$ with a worker $k$. 
Communications between the workers are governed by the structure of the linear operators $(\Db_\idcomposite)_{1 \le \idcomposite \le \ncomposite}$. Precisely, the notion of hypergraph will be used to deal with this aspect. 
\Cref{Sec:Model} is dedicated to technical details.
For the sake of clarity and simplicity of the presentation, the main part of the paper presents the proposed distributed algorithm {with $\ncomposite = 1$ in~\eqref{eq:dist-gen-split}, corresponding to the target} distribution~\eqref{eq:dist-gen-lin}. 
The appendix extends the proposed approach to distributions of the form~\eqref{eq:dist-gen-split} with $\ncomposite > 1$. 
We also focus on the case with augmentation variables as in~\eqref{eq:dist-SPA} to ensure better mixing properties~\cite{Vono_etal_2019}. Note that our approach can be used without this optional variable.

The proposed algorithm accommodates several distributed architectures, and is particularly suitable for a Single Program Multiple Data (SPMD) architecture~\cite{Darema2001}. In contrast with a client-server configuration, all the workers involved in an SPMD architecture execute similar tasks on a subpart of all the variables, with no central server. An SPMD architecture can drastically reduce the communication costs compared to a client-server architecture when a small number of workers is involved in each communication channel. This is especially the case when the structure of $\Db$ induces \emph{localized} couplings between parameters. Figure~\ref{fig:blocksparse} illustrates the particular case of a block-sparse matrix $\Db$: localized interactions between parameters imply a block-sparse structure which induces a hypergraph structure, see \Cref{Sec:Model} for details. This case is often encountered in practice, e.g., for inverse problems in imaging, with applications such as image deconvolution or inpainting, or when considering models based on a TV norm regularization~\cite{Rudin1992}. 
When multiple composite terms are considered as in~\eqref{eq:min-gen-multi}, a parallel implementation based on a client-server architecture would be possible by exploiting conditional independence between blocks of variables. However this configuration may suffer from communication bottlenecks, with a number of workers limited by the number of conditionally independent blocks of variables in the model. In contrast, the SPMD approach permits to use a larger number of workers by exploiting the structure of the hypergraphs induced by the structure of operators $\Db_\idcomposite$ involved.

The remainder of the paper is organized as follows. The AXDA approach and the SGS algorithm~\cite{Vono_etal_2019, Vono_etal_2020} are recalled in~\Cref{Sec:SOA}. 
An overview of the proposed distributed SGS is given in~\Cref{Sec:dist}. This section further outlines how limitations of the client-server architecture can be addressed by a fully decentralized SPMD architecture. The advantages offered by the latter are specifically emphasized for distributions defined on a high-dimensional space, e.g., in imaging inverse problems.
The proposed hypergraph model is introduced in \Cref{Sec:Model}, and the associated distributed block-coordinate SGS is given in \Cref{Sec:dist-single-algo}. \Cref{Sec:application} describes an SPMD implementation of the proposed method for a large scale image deconvolution problem. Conclusion and perspectives are given in~\Cref{Sec:conclusion}. 
Eventually, the proposed method is extended to the general case of distributions involving multiple composite terms in
\if0\production
{~\Cref{appendix:SPA-PSGLA-multi}.} 
\else
{~the online appendix.} 
\fi

\section{AXDA approach and the split Gibbs samplers}
\label{Sec:SOA}
In this section, we summarize the approximation results at the basis of the AXDA approach~\cite{Vono_etal_2019}. We also describe a Gibbs sampler to approximately draw samples from~\eqref{eq:dist-gen-lin}, referred to as the Split Gibbs sampler (SGS)~\cite{Vono_etal_2020}.

\subsection{AXDA splitting approach}
\label{Ssec:SPA-simple}

A Gibbs sampler can be used to draw samples from \eqref{eq:dist-SPA}, approximating the target distribution~\eqref{eq:dist-gen-lin}.
The sampler successively draws one sample from each conditional distribution

\begingroup
\allowdisplaybreaks
\begin{align} 
    \distc{\alpha}{\xb}{\zb,\ub} 
    & \propto \exp \big( -h(\xb) -f(\xb)  - \phi_\alpha(\Db\xb,\zb-\ub) \big) ,    \label{eq:dist-SPA-cond:x}\\
    \distc{\alpha}{\zb }{ \vb,\ub }  
    & \propto \exp \big( -g(\zb) - \phi_\alpha(\vb,\zb-\ub)  \big) ,    \label{eq:dist-SPA-cond:z}\\
    \distc{(\alpha,\beta)}{\ub }{ \vb,\zb } 
    & \propto \exp \big( - \phi_\alpha(\vb,\zb-\ub) - \psi_\beta(\ub) \big) , \label{eq:dist-SPA-cond:u}
\end{align}
\endgroup
where $\vb = \Db \xb$.
The SGS associated with~\eqref{eq:dist-SPA-cond:x}--\eqref{eq:dist-SPA-cond:u} is reported in~Algorithm~\ref{algo:SPA-simple}.

\begin{algorithm}[htbp]
    \small
    \KwIn{$\zb^{(0)} \in \GGG$, $\ub^{(0)} \in \GGG$, $(\alpha, \beta) \in ]0,+\infty[^2$}
    \For{$t = 0$ \KwTo $T$}{
        $\xb^{(t)} \sim \distc{\alpha}{\xb }{ \zb^{(t)},\ub^{(t)} }$, \\
        $\vb^{(t)} = \Db \xb^{(t)}$, \\
        $\zb^{(t+1)} \sim \distc{\alpha}{\zb }{ \vb^{(t)}, \ub^{(t)} }$, \\
        $\ub^{(t+1)} \sim \distc{(\alpha,\beta)}{\ub}{ \vb^{(t)}, \zb^{(t+1)}}$
    }
    \KwOut{$(\xb^{(t)})_{1 \le t \le T}$, $(\zb^{(t)})_{1 \le t \le T}$, $(\ub^{(t)})_{1 \le t \le T}$}
    \caption{\small Generic Split Gibbs Sampler (SGS)~\cite{Vono_etal_2019}.}
    \label{algo:SPA-simple}
\end{algorithm}

Under technical conditions on the functions $\phi_\alpha$ and $\psi_\beta$ in \eqref{eq:dist-SPA}, \cite{Vono_etal_2019} showed that the marginal distribution of $\xb$ can be made arbitrarily close to the original distribution $\dist{}{\xb}$ given in~\eqref{eq:dist-gen-lin}, typically when $(\alpha, \beta)\to (0,0)$. This result is summarized in the following proposition.
\begin{proposition} \label{Prop:cv-SPA-simple}\cite[Thm. 1 and 2]{Vono_etal_2019}
Let $(\alpha ,\beta) \in \RPP^2$, $\phi_{\alpha} \colon \GGG \times \GGG \to ]-\infty, +\infty]$ and $\psi_{\beta} \colon \GGG \to ]-\infty, +\infty]$. 
Assume that
\begin{equation} \label{prop:cv-spa-simple-cond1}
    (\forall \xb \in \HHH) \quad \lim_{\alpha \rightarrow 0} \frac{\exp\big(- \phi_{\alpha}(\Db \xb, \zb)\big)}{\int_{\GGG} \exp\big(- \phi_{\alpha}( \Db \xb, \zb)\big) d\zb} = \delta_{\xb} (\zb),
\end{equation}
and that, there exists $\eta_{\alpha,\beta} \in \RPP$,
\begin{equation}\label{prop:cv-spa-simple-cond2}
    (\forall (\xb, \zb) \in \HHH \times \GGG) \quad 
    \int_{\GGG} \exp\big(- \phi_{\alpha}(\Db \xb, \zb-\ub) - \psi_{\beta}(\ub)  \big) d\ub 
    \propto \exp\big(-\phi_{\eta_{\alpha,\beta}}(\Db \xb, \zb)\big).
\end{equation}
Let, for $\xb \in \HHH$, $\pi_{\eta_{\alpha,\beta}}(\xb) = \int_{\GGG \times \GGG}  \pi_{\alpha, \beta} (\xb, \zb, \ub) d\ub d\zb$.
Then, $\|\pi - \pi_{\eta_{\alpha, \beta}}\|_{\text{TV}} \to 0$ when $\eta_{\alpha,\beta} \to 0$. 
\end{proposition}

\Cref{Prop:cv-SPA-simple} ensures that the marginal distribution of $\xb$ with respect to $\pi_{(\alpha, \beta)}$ converges to the target distribution $\pi$ given in~\eqref{eq:dist-gen-lin}, provided conditions~\eqref{prop:cv-spa-simple-cond1} and \eqref{prop:cv-spa-simple-cond2} hold. 
In particular, in \cite{Vono_etal_2019}, these conditions are shown to be satisfied for $\eta_{\alpha, \beta}^2 = \alpha^2 + \beta^2$, with
\begin{equation}    \label{eq:phi-psi-def}
    \begin{cases}
    \phi_\alpha(\Db \xb, \zb-\ub ) = \frac{1}{2 \alpha^2} \| \Db\xb - (\zb-\ub) \|^2, \\
    \psi_\beta(\ub) = \frac{1}{2 \beta^2} \| \ub \|^2.
    \end{cases}
\end{equation}
In the particular case when $f\equiv 0$, stronger results can be found in \cite[Theorem 2]{Vono_etal_2020}, including theoretical guarantees on the convergence rate.

\subsection{PSGLA within SGS}
\label{Ssec:SPA-PSGLA}

Drawing samples directly from the conditional distributions \eqref{eq:dist-SPA-cond:x}--\eqref{eq:dist-SPA-cond:u} can still be difficult. This is the case for the application described in \Cref{Sec:application}. 
To overcome this issue, samples can be drawn using appropriate transition kernels. 

To sample from~\eqref{eq:dist-SPA}, Metropolis-Hastings transition kernels are classical choices to draw samples from the conditional distributions $\distc{\alpha}{\xb}{\zb,\ub}$, $\distc{\alpha}{\zb }{ \xb, \ub }$ and $\distc{(\alpha,\beta)}{\ub}{ \xb, \zb }$, leading to a conditional Metropolis-Hastings sampler~\cite{Jones2014}. Appropriate proposals include Langevin-based kernels, such as the Moreau-Yosida unadjusted Langevin algorithm (MYULA)~\cite{Durmus2018} and the proximal stochastic gradient Langevin algorithm (PSGLA) \cite{Salim_Richtarik_2020}. These kernels can handle differentiable and non-differentiable potential functions simultaneously, and have been shown suitable to address high dimensional problems. To avoid the extra cost of the accept-reject step, approximate sampling can be considered while maintaining a good approximation of the target distribution by using unadjusted kernels, as in~\cite{Durmus2018}.
This approach is adopted in the following. Deriving non-asymptotic convergence bounds and analyzing the bias between the target distribution and~\eqref{eq:dist-SPA} is however beyond the scope of this paper.

The choice of suitable transition kernels will be instrumental to design a distributed Gibbs sampler when direct sampling from the conditional distributions is difficult. 
Technical assumptions on the functions in~\eqref{eq:dist-SPA} will ensure that the transition kernels are amenable to a distributed implementation, see \Cref{Sec:Model}. This implementation will be adapted to the specific structure of the distributions $\distc{\alpha}{\xb}{ \zb,\ub}$, $\distc{\alpha}{\zb}{ \xb,\ub}$ and $\distc{\alpha,\beta}{\ub}{ \xb,\zb}$.
In particular, for the proposed distributed SGS, we consider the case where any appropriate transition kernel can be used to sample $(\zb, \ub)$. The parameter $\xb$ is approximately sampled using a PSGLA transition to avoid the extra cost induced by a Metropolis correction step. To this aim, we assume that $h$ is $\lambda_h$-Lipschitz-differentiable. 
For simplicity, $\phi_\alpha$ and $\psi_\beta$ are also taken as $\ell_2$ norms as in~\eqref{eq:phi-psi-def}, so that \eqref{prop:cv-spa-simple-cond1} and \eqref{prop:cv-spa-simple-cond2} are satisfied. Then $\ub$ can be directly sampled from its Gaussian conditional distribution. These blanket assumptions, summarized in~\Cref{ass:psgla-spa}, will be adopted in the following.
\begin{assumption} \label{ass:psgla-spa}\
    \begin{enumerate}
        \item\label{ass:psgla-spa:phi_psi}
            $\phi_\alpha$ and $\psi_\beta$ are given by
            \begin{align}
            (\forall (\vb, \ub) \in \GGG^2)\quad
            &  \phi_\alpha(\vb, \ub) = \frac{1}{2\alpha^2} \| \vb - \ub \|^2, \label{eq:phi} \\
            &   \psi_\beta(\ub) = \frac{1}{2\beta^2} \| \ub \|^2.  \label{eq:psi}
            \end{align}
        \item\label{ass:psgla-spa:smoothness}
            $h \colon \HHH \to \RR$ is $\lambda_h$-Lipschitz differentiable, with $\lambda_h>0$.
    \end{enumerate}
\end{assumption}

\begin{algorithm}[htbp]
    \small
    \KwIn{$\xb^{(0)} \in \HHH$, $(\zb^{(0)}, \ub^{(0)}) \in \GGG^2$, $(\alpha, \beta) \in ]0,+\infty[^2$, $\gamma \in \big]0, \big(\lambda_h + \| \Db \|^2/\alpha^2\big)^{-1} \big[$}

    $\vb^{(0)} = \Db \xb^{(0)}$; \\

    \For{$t = 0$ \KwTo $T$}{
        \tcp{Draw $\xb^{(t+1)}$ with a PSGLA~\cite{Salim_Richtarik_2020} kernel}
        $\xb^{(t+1)} = \prox_{\gamma f} \bigg( \xb^{(t)} - \gamma \nabla h(\xb^{(t)}) - \gamma \Db^* \Big( \nabla \phi_\alpha( \cdot , \zb^{(t)}-\ub^{(t)}) (\vb^{(t)}) \Big) + \sqrt{2\gamma} \, \wb^{(t)} \bigg)$; \\
        $\vb^{(t+1)} = \Db \xb^{(t+1)}$; \\
        \tcp{Draw $\zb^{(t+1)}$ and $\ub^{(t+1)}$ from their conditional distribution}
        $\zb^{(t+1)} \sim \distc{\alpha}{\zb}{\vb^{(t+1)}, \ub^{(t)}}$; \\
        $\ub^{(t+1)} \sim \Nc \Big( \frac{\beta^2}{\alpha^2+\beta^2} (\zb^{(t+1)} - \vb^{(t+1)}), \frac{\alpha^2+\beta^2}{\alpha^2 \beta^2} \Idb \Big)$;
    }
    \KwOut{$(\xb^{(t)})_{1 \le t \le T}$, $(\zb^{(t)})_{1 \le t \le T}$, $(\ub^{(t)})_{1 \le t \le T}$}
    \caption{\small Proposed sampler (PSGLA within SGS).}
    \label{algo:SPA-simple-sampling}
\end{algorithm}


\Cref{ass:psgla-spa}\ref{ass:psgla-spa:phi_psi} ensures that \Cref{Prop:cv-SPA-simple} is valid, see~\eqref{eq:phi-psi-def}. \Cref{ass:psgla-spa}\ref{ass:psgla-spa:smoothness} is necessary to use the PSGLA transition kernel in the proposed SGS detailed in Algorithm~\ref{algo:SPA-simple-sampling},
where $\Db^*$ denotes the adjoint of $\Db$, and $(\wb^{(t)})_{1 \le t \le T}$ is a sequence of independent and identically distributed (i.i.d) standard Gaussian random variables in~$\HHH $.

\section{Proposed distributed SGS in a nutshell}
\label{Sec:dist}
For high-dimensional problems, every step of Algorithm~\ref{algo:SPA-simple} is computationally expensive, calling for a distributed algorithm. Steps 5 and 6 are easy to distribute since they correspond to term-wise operations. In contrast, steps 1, 3 and 4 require communications due to the presence of the linear operator $\Db$ and its adjoint $\Db^*$. This work is aimed at designing an efficient distributed version of Algorithm~\ref{algo:SPA-simple-sampling} by exploiting the hypergraph structure that can emerge from the topology of $\Db$. 
This section describes the proposed approach \textit{in a nutshell}, reducing technicality. 
Section~\ref{Sec:Model} introduces the model to describe the hypergraph structure and the resulting distribution strategy. Section~\ref{Sec:dist-single-algo}  details the distributed version of the proposed Algorithm~\ref{algo:single-synch}.

In \Cref{Ssec:dist:client-server-vs-spmd}, we first explain the advantages offered by an SPMD distributed implementation over a client-server approach to address high dimensional problems.
To enable this SPMD implementation, we rely on hypergraph structures that facilitate flexible communications. These hypergraph structures are reminiscent of the topology of $\Db$, as described in \Cref{Ssec:dist:hypergraph}.  
Practical considerations for an efficient distributed implementation are given in~\Cref{Sssec:Model:cond-eff-dist}. 
We refer the reader to \Cref{Sec:Model} and \Cref{Sec:dist-single-algo} for a rigorous description of the proposed distributed algorithm. For an example application of the proposed approach to an imaging inverse problem, the reader can jump to \Cref{Sec:application}.

\subsection{Splitting-based samplers} 
\label{Ssec:dist:client-server-vs-spmd}

Splitting-based samplers from the literature~\cite{Vono_etal_2020,Rendell2021} can benefit from a distributed implementation on a client-server architecture. These algorithms can accommodate densities of the form~\eqref{eq:dist-gen-split-SPA} with $\ncomposite > 1$ linear operators. After splitting, a collection of $K - 1 \leq \ncomposite$ workers (i.e., \emph{clients}) handle computations associated with groups of variables conditionally independent from one worker to another. The server handles operations on the full variable $\xb$, possibly duplicated across all the workers. This configuration can suffer from several drawbacks listed below.

\begin{enumerate}
    \item \textbf{Shared variable and communication costs} :
    The computation of $(\Db_\idcomposite\xb)_{1 \leq \idcomposite \leq \ncomposite}$ on the server can be expensive. The shared parameter $\xb$ needs to be broadcast to all the clients at each iteration. Communications can significantly increase for applications defined on a high dimensional parameter space, such as inverse problems in imaging applications.
    \item \textbf{Limited number of workers} :
    The number of workers $K$ is restricted by the number of composite functions $\ncomposite$ since $K-1 \leq \ncomposite$, typically with $1 \leq \ncomposite \leq 5$ for most imaging inverse problems. A client-server approach cannot be applied to distributions with $\ncomposite = 1$ such as~\eqref{eq:dist-gen-lin}.
    \item \textbf{Load balancing}: 
    Computing costs induced by the operators $(\Db_\idcomposite)_{1 \leq \idcomposite \leq \ncomposite}$ can be hard to balance over the workers. Prohibitive idle time due to synchronization with \emph{stragglers} may thus drastically limit the parallel efficiency of the algorithm.
\end{enumerate}

To address these issues, a possible approach consists in exploiting the hypergraph structure of the operator $\Db$, as outlined in \Cref{Ssec:dist:hypergraph}.
Such an approach enables the use of an SPMD architecture~\cite{Darema2001}, so that all the workers can conduct locally the same tasks over a subpart of the variables only. In this configuration, load balancing is in general easier to handle, with no restriction on the number of workers $K$ imposed by the model and lighter communications as explained in \Cref{Ssec:dist:hypergraph,Sssec:Model:cond-eff-dist}. 
\Cref{Ssec:algo:dist-spmd} explains in more details how the SPMD architecture takes advantage of the hypergraph structure of the algorithm detailed in \Cref{Sec:Model,Sec:dist-single-algo}.

\subsection{A hypergraph structure to better communicate} 
\label{Ssec:dist:hypergraph}

We consider a localized linear operator $\Db$ in~\eqref{eq:dist-gen-lin}, such that couplings between latent parameters are localized. This is for instance the case when $\Db$ is block-sparse, see \Cref{fig:blocksparse}. 
Such a structure is not strictly necessary to design the proposed distributed algorithm but will be instrumental in practice to reduce the communications between workers for an efficient implementation of the algorithm. 
Practical uses of interest for block-sparse structures include imaging inverse problems, e.g., where $\Db$ corresponds to a deconvolution or inpainting measurement operator. Block-sparse operators also appear when considering models such as the TV norm.

The structure of the matrix $\Db$ can be described with a binary matrix that can be interpreted as the adjacency matrix of a hypergraph. 
A hypergraph is a generalization of a graph, in which an edge can join any finite number of vertices. Formally, an undirected hypergraph $\Hc$ is a pair $(\xb, \eb)$ made of vertices $\xb = (x_n)_{1 \le n \le N}$ and hyperedges $\eb = (e_m)_{1\le m \le M}$, where each hyperedge is a set of connected vertices, that is a subset of the $(x_n)_{1 \le n \le N}$.
The main idea is to take benefit from this structure to distribute the computation of various quantities in Algorithm~\ref{algo:SPA-simple-sampling} over $K$ workers.
It appears that the product $\Db\xb$ is of special interest.

\Cref{fig:blocksparse} illustrates the proposed approach.
The variable $\xb$ is divided into $K$ blocks $(\xb_k)_{1 \le k \le K}$, with each block $\xb_k$ assigned to a single worker.
The hyperedges will characterize the necessary communications between the workers. A worker $k \in \{1, \ldots, K\}$ stores $\xb_k$ but it may as well need  access to some coefficients of $\xb_{k'}$ stored on another worker $k'\neq k$ to carry out its computations.
Note that the operator $\Db$ will be split once and for all in $\Db=(D_{m,n})_{1 \le m \le M, \, 1 \le n \le N}$  in an adequate manner over the $K$ workers as well; this decomposition will not necessitate any additional communication between the workers.
%
The output quantity $\vb = (v_m)_{1 \le m \le M} = \Db\xb $ computed in Algorithm~\ref{algo:SPA-simple-sampling} will also be distributed among the $K$ workers. 
To this aim, each hyperedge $m\in \{1, \ldots, M\}$ is associated to a worker denoted by $k_m \in \{1, \ldots, K\}$.
Then, each subpart $v_m$ will be computed and stored on worker $k_m$. 
In practice, only the vertices $x_n$ that correspond to non-zero blocks $D_{m, n}$ in $\Db$ will be necessary to compute $v_m$.
Communications occur when the worker $k=k_m$ requires a subpart of $\xb_{k'}$ stored on another worker $k'\neq k_m$ to carry out the computation of $v_m$. 
A subset of each $\xb_k$ will never be involved in communications since it is used to compute $v_m$ only, and will not be used to compute any other $v_{m'}$. 
Finally, the worker $k=k_m$, that is in charge of the computation of $v_m$, will also store $\xb_k$ to reduce communications.
In the particular case when $\Db$ is block-sparse, the communication cost will remain small as long as the subpart of $\xb_{k'}$ required by worker $k_m$ remains small.

\subsection{Conditions for an efficient distributed implementation}
\label{Sssec:Model:cond-eff-dist}

Even though the proposed approach detailed in \Cref{Sec:Model,Sec:dist-single-algo} is very general, some particular structures of matrix $\Db$ can ensure limited communications, which are often the bottleneck of distributed methods.

Complexity costs of algorithms can roughly be divided into three terms. First, a computation term $\tflop$ reflects the time to perform a single floating point operation. Second, a communication term $\tband$, defined as the inverse of the communication bandwidth, quantifies the time necessary to send a single value. Third, a latency term $\tlatency$ represents the cost incurred by establishing a communication. 
In practice, $\tflop \ll \tband \ll \tlatency$. This implies that the number of communications and the size of the messages need to be sufficiently small to ensure the overall communication time to be negligible compared to the computation time.

Consequently, to ensure a higher efficiency of the proposed distributed SGS, both in terms of computations and communications, we can identify two conditions on the structure of $\Db$.
These conditions are typically satisfied when $\Db$ is block-sparse, as the operators considered in \Cref{Sec:application}.

The first condition consists in ensuring that, each worker only needs to communicate with a small amount of other workers compared to the total number of workers $K$.

The second condition is that the number of variables that two communicating workers need to exchange remains small compared to the variables already stored on each worker. In particular, for two workers $(k, k')\in \{1, \ldots, K\}^2$, $k\neq k'$, the subpart of $\xb_k$ that need to be sent from $k$ to $k'$ needs to be small compared to $\xb_k$ and $\xb_{k'}$.

\section{Model and hypergraph structure}
\label{Sec:Model}
This section provides the notation and a description of the hypergraph model used to define the proposed distributed SGS. It gives a formal description of the proposed distributed computing strategy, which exploits the structure of the linear operator $\Db$ and a separability assumption on the functions involved in~\eqref{eq:dist-SPA}.
Notation is summarized in Tables~\ref{table:notation-hypergraph} and~\ref{table:notation2-hypergraph}, and illustrated on a simple example in \Cref{fig:opsplitting}.

\begin{table}[!ht]
    \centering\footnotesize
    \begin{tabular}{@{}p{4cm}@{}|p{8.5cm}|@{}c@{}|@{}c@{}} \toprule 
        \multirow{2}{*}{Notation} & \multirow{2}{*}{Definition} & \multicolumn{2}{c}{Given by}    \\
        && operator $\Db$ & $\phantom{a}$ user $\phantom{a}$ \\
        \hline
        $n\in \{1, \ldots, N\}$ & Indices for vertices & \checkmark &  \\
        $m\in \{1, \ldots, M\}$ & Indices for hyperedges & \checkmark & \\
        $\eb = (e_m)_{1 \le m \le M}$ & Hyperedges & \checkmark & \\
        $e_m \subset \{1, \ldots, N\}$ & Vertex indices in hyperedge $m$ & \checkmark & \\
        $k \in \{1, \ldots, K\}$ & Indices for workers & & \checkmark \\
        $\Rc_k \subset \{1, \ldots, K\}\setminus\{k\}$ & Indices of workers $k' \in \{1, \ldots, K\}\setminus\{k\}$ from which worker~$k$ receives vertex values & & \checkmark \\
        $\Sc_k \subset \{1, \ldots, K\}\setminus\{k\}$ & Indices of workers $k' \in \{1, \ldots, K\}\setminus\{k\}$  to which worker~$k$ sends vertex values & & \checkmark \\
        $\eV_k \subset \{1, \ldots, N\}$ & Indices of vertices stored on worker~$k$, such that $(\eV_k)_{1 \le k \le K}$ is a partition of $\{1, \ldots, N\}$ & \checkmark & \checkmark \\
        $k_m \in\{1, \ldots, K\}$ & Worker associated with $m$-th hyperedge $e_m$ (chosen by the user). $k_m$ must satisfy $e_m \cap \eV_{k_m} \neq \emp$  & & \checkmark \\
        $\Wc_m \subset \{1, \ldots, K\} \setminus \{k_m\}$ & Set of all workers but $k_m$, containing vertices from $e_m$ & \checkmark & \\
        $\overline{\Wc}_m \subset \{1, \ldots, K\} $ & $\overline{\Wc}_m = k_m \cup \Wc_m$ & \checkmark & \\
        $\eV_{(k,k')} \subset \eV_{k'}$ & Indices of vertices sent from worker $k'$ to worker~$k$ & \checkmark & \checkmark    \\
        $\eV_{\Wc_m} \subset \{1, \ldots, N\}\setminus \eV_{k_m}$ & $\eV_{\Wc_m} = \bigcup_{k' \in \Wc_m} \eV_{(k_m, k')}$ the set of vertex indices that will be communicated to worker $k_m$ from all workers $k'\in \Wc_m$ & \checkmark & \checkmark    \\
        $\eV_{\overline{\Wc}_m} \subset \{1, \ldots, N\}$ & $\eV_{\overline{\Wc}_m} = \eV_{k_m} \cup \eV_{\Wc_m}$ the set of vertex indices necessary to perform computations associated with $k_m$ & \checkmark & \checkmark    \\
        $\eE_k \subset \{1, \ldots, M\}$ & Indices of hyperedges only containing vertices stored on worker~$k$ & \checkmark & \checkmark \\
        $\eE_{(k,k')} \subset \{1, \ldots, M\}$ & Indices of hyperedges containing vertices sent from worker $k'$ to worker~$k$ & \checkmark & \checkmark \\
        $\eE_{\Rc_k} \subset \{1, \ldots, M\}$ & $\eE_{\Rc_k} = \cup_{k'\in \Rc_k} \eE_{(k,k')}$ set of all hyperedges containing vertices that will be communicated to worker $k$ & \checkmark & \checkmark    \\
        $\overline{\eE}_k \subset \{1, \ldots, M\}$ & $ \overline{\eE}_k = \eE_k \cup \eE_{\Rc_k} $, such that $(\overline{\eE}_k)_{1 \le k \le K}$ is a partition of $\{1, \ldots, M\}$ & \checkmark & \checkmark \\
        \bottomrule
    \end{tabular}
    \caption{\label{table:notation-hypergraph}\small
    Sets used to define the hypergraph structure of the problem and the distributed algorithm. 
    }
\end{table}

\begin{table}[h!]
    \centering\footnotesize
    \begin{tabular}{@{}p{4.5cm}@{}|@{}p{11.1cm}@{}} \toprule
        Notation & Definition  \\
        \hline
        $\xb = (x_n)_{1 \le n \le N} \in \HHH$ & Vertex values \\
        $\Db \colon \HHH \to \GGG$ & $\Db = (D_{m,n})_{1 \le m \le M, 1\le n \le N}$ linear operator defining the hypergraph structure \\
        $\ub$, $\vb\in \GGG$ & \multirow{2}{*}{Hyperedge weights} \\
        $\vb = (v_m)_{1 \le m \le M} \in \GGG$ & \\
        $\vb_k \in \GGG_k$ & $\vb_k = (v_m)_{m \in \overline{\eE}_k}$ hyperedge weights stored on worker $k$ \\        
        $\xb_k  \in \HHH_k$ & $\xb_k = (x_n)_{n \in \eV_k}$ vertex values stored on worker $k$ \\
        $\Db_{m,k} \colon \HHH_k \to 
        \GG_m $ & $\Db_{m,k} = (D_{m,n})_{n \in \eV_k}$, for $m\in \eE_k$, subpart of $\Db$ stored on worker $k$, associated with hyperedges containing vertices only on worker $k$ \\
        $\xb_{(k,k')} \in \underset{n\in \eV_{(k,k')}}{\Cart} \HH_n$  & $\xb_{(k,k')} = (x_n)_{n \in \eV_{(k,k')}}$ vertex values sent from worker $k'$ to worker $k$ \\
        $\Db_{m,(k,k')} \colon \underset{n\in \eV_{(k,k')}}{\Cart} \HH_n \to
        \GG_m $ & $\Db_{m,(k,k')} = (D_{m,n})_{n \in \eV_{(k,k')}}$, for $m\in \eE_{(k,k')}$, subpart of $\Db$ stored on worker $k$, associated with hyperedges containing vertices overlapping workers $k$ and $k'$ \\
        $\xb_{\overline{\Wc}_m} \in \underset{n\in \eV_{\overline{\Wc}_m}}{\Cart} \HH_n$  & $\xb_{\overline{\Wc}_m} = (x_n)_{n \in \eV_{\overline{\Wc}_m}}$, concatenation of vertex values stored on worker $k_m$, and those sent from all worker $k' \in \Wc_m$ to worker $k_m$ \\
        $\Db_{m,\overline{\Wc}_m} \colon \underset{n\in \eV_{\overline{\Wc}_m} }{\Cart} \HH_n \to
        \GG_m $ &
        $\Db_{m,\overline{\Wc}_m} = (D_{m,n})_{n \in \eV_{\overline{\Wc}_m}}$,
        for $m\in \eE_{\Rc_k}$, subpart of $(D_{m,n})_{1 \le n \le N}$ corresponding to vertices stored either on worker $k_m$ or on a worker in $\Wc_m$ \\
        \bottomrule
    \end{tabular}
    \caption{\label{table:notation2-hypergraph}\small
    Notation used for the variables involved in the proposed distributed algorithm. 
    }
\end{table}

\subsection{Hypergraph structure of the model}
\label{Ssec:hypergraph-struct-model}

Let $\HHH = \RR^{\overline{N}}$ be such that $\HHH = \HH_1 \times \ldots \times \HH_N$, where for every $n\in \{1, \ldots, N\}$, $\HH_n = \RR^{N_n}$ and $\overline{N} = \sum_{n=1}^N N_n$. 
An element of $\HHH$ is denoted by $\xb = (x_n)_{1 \le n \le N}$, where, for every $n\in \{1, \ldots, N\}$, $x_n \in \HH_n$. 
Similarly, let $\GGG = \RR^{\overline{M}}$ be such that $\GGG = \GG_1 \times \ldots \times \GG_M$, where for every $m\in \{1, \ldots, M\}$, $\GG_m = \RR^{M_m}$ and $\overline{M} = \sum_{m=1}^M M_m$. Let $\vb = (v_m)_{1 \le m \le M}$ be an element of $\GGG$ such that, for every $m \in \{1, \ldots, M\}$, $v_m \in \GG_m$.

The distribution of the problem over the different workers will follow the topology of a hypergrah encoded by the structure of $\Db$. To this aim, we assume that $\Db$ holds some block separability structure, and consider the following hypergraph model.
\begin{model}\label{mod:hypergraph}
Let $\Db \colon \HHH \to \GGG$ be such that $\Db = (D_{m,n})_{1 \le m \le M, 1\le n \le N}$, with, for every $(m,n) \in \{1, \ldots,M\} \times \{1, \ldots,N\}$, $D_{m,n} \colon \HH_n \to \GG_m$.
Let $\Hc = (\xb, \eb)$ be the hypergraph associated with $\Db$, with $N$ vertices $\xb$, and $M $ hyperedges denoted by $\eb =(e_m)_{1\le m \le M}$, such that 
\begin{equation}
    (\forall m \in \{1, \ldots, M\}) \quad
    e_m = \menge{ n \in \{1, \ldots, N\}}{ D_{m,n} \neq 0},
\end{equation}
where $0_{\GG_m \times \HH_n}$ is the null element from $\GG_m$ to $\HH_n$
\end{model}
The hyperedges of $\Hc$, i.e., the connections between vertices, are described by the topology of $\Db$. Precisely, the $M$ rows of $\Db$ represent the $M$ hyperedges of $\Hc$, and the $N$ columns represent the $N$ vertices of $\Hc$. Hence, as described in \Cref{mod:hypergraph}, for each $m\in \{1, \ldots, M\}$, the hyperedge $e_m$ links nodes $(x_n)$ if $D_{m,n} \neq 0_{\GG_m \times \HH_n}$.
Hence, $\Db$ can be seen as a weighted incidence matrix associated with $\Hc$.
In this context, any variable in $ \GGG$ can be seen as a hyperedge weight of the hypergraph $\Hc$.

According to \Cref{mod:hypergraph}, the computation of $\vb = \Db \xb$ in Algorithm~\ref{algo:SPA-simple-sampling} can be decomposed and computed block-wise as
\begin{equation}
\label{eq:vm-Dmn-xn}
    \vb = (v_m)_{1 \le m \le M} \text{ with }
    (\forall m \in \{1, \ldots, M\}) \quad
    v_m = \sum_{n\in e_m} D_{m,n} x_n.
\end{equation}
In~\eqref{eq:vm-Dmn-xn}, for every $m\in \{1, \ldots, M\}$, only the non-zero blocks of $\Db$ are taken into account, i.e., only using the vertices contained in the hyperedge $e_m$.
In the particular case when all hyperedges are disjoint (i.e., disconnected hypergraph), a very simple distributed algorithm would distribute the computation of the quantities $v_m$ on independent workers. This is not the case in general, and hyperedges $e_m$ can share some vertices $x_n$. A distributed algorithm will thus need to carry out communications between some workers. The next section exploits the hypergraph model described here to distribute the computation of $\Db\xb$, while limiting the communication cost.

\subsection{Distribution of the hypergraph nodes over the workers}
\label{Ssec:Model:dist-hyp-notation}

Let $K \in \NN^*$, $K \le N$, be the number of workers available to the user to parallelize the algorithm. 
On each worker $k \in \{1, \ldots, K\}$, we store a subpart $\xb_k$ of the nodes $\xb$ of the hypergraph $\Hc$. 
The pattern for the distribution of the nodes $\xb$ over the $K$ workers will be driven by the topology of the matrix~$\Db$, as explained in \Cref{Ssec:hypergraph-struct-model}. This pattern needs to be fixed before designing the algorithm.
Note that the blocks $(D_{m,n})_{1 \le m \le M, \, n \in e_m}$ of $\Db$ will be distributed over the workers as well, but will not require to be communicated. 
In particular, we provide below the notation to split the hypergraph (i.e., the operator $\Db$) over the $K$ workers. The notation for the splitting of the nodes and the hyperedge weights are illustrated in \Cref{fig:opsplitting} for a simple example of a block-sparse matrix $\Db$.

\subsubsection{Distributing the hypergraph over the workers}
\label{Sssec:distributed-hypergraph}

Let $(\eV_k)_{1 \le k \le K}$ be the partition of the set of vertex indices $\{1, \ldots, N\}$ such that, for every $k\in \{1, \ldots, K\}$, $\eV_k $ is the set of vertex indices handled on worker $k$, and for every $\xb \in \HHH$,  $\xb_k = (x_n)_{n \in \eV_k}$.
Recall that $k_m \in \{1, \ldots, K\}$ denotes the worker associated with the $m$-th hyperedge $e_m$. The association of hyperedges to workers is a choice left to the user. Then, we must have  $e_m \cap \eV_{k_m} \neq \emp$, i.e., worker $k_m$ must contain at least one vertex belonging to hyperedge $e_m$. 

The set of all the hyperedges associated with worker $k$ is denoted by
\begin{equation}    \label{def:Ebar}
\overline{\eE}_k = \menge{m \in \{1, \ldots, M\}}{k_m = k}.
\end{equation} 
Note that $\big( \overline{\eE}_k \big)_{1 \leq k \leq K}$ defines a partition of $\{1,\ldots, M\}$ over the $K$ workers.
Optimizing the configuration for a specific operator $\Db$ is, on its own, a resource allocation problem~\cite{Pilla2021} that is out of the scope of this work.
Using notation~\eqref{def:Ebar}, for any $\vb \in \GGG$, we can denote by $\vb_k = (v_m)_{m\in \overline{\eE}_k}$ the hyperegde weights stored on worker $k$.

We also define $\eE_k$, the set of hyperedge indices only containing vertices stored on worker $k$, as
\begin{equation}    \label{eq:def:Ek}
    \eE_k = \menge{ m \in \{1, \ldots, M\} }{ n \in e_m \; \Leftrightarrow \; n \in \eV_k}.
\end{equation}
The sets $(\eE_k)_{1 \le k \le K}$ identify the rows of $\Db$ whose non-zero elements are multiplied with vertices that are stored on a single worker, i.e., the rows that can be used with no communication between two different workers. In contrast, the computation of $v_m$ for $m\in\overline{\eE}_k\setminus\eE_k$ will call for communications.

\subsubsection{Hyperedges overlapping over workers ensuring communications}
\label{Ssec:communication-overlap-hyperedge}

The hyperedges overlapping over multiple workers will require vertices stored on different workers to be communicated between each other.
%
They correspond to hyperedges $e_m$, with $ m \in \{1, \ldots, M\} \setminus (\bigcup_{k=1}^K \eE_k)$.

For a fixed hyperedge $e_m$, let $\Wc_m \subset \{1 , \ldots, K\}\setminus \{k_m\}$ be the set of workers different from $k_m$ containing vertices $x_n$ belonging to the same hyperedge $e_m$, i.e., 
\begin{equation*}
\Wc_m = \menge{k \in \{1 , \ldots, K\}\setminus \{k_m\} }{ \exists n \in e_m \text{ such that } n \in \eV_{k}}.
\end{equation*}
Using this notation, we can give an equivalent definition of \eqref{eq:def:Ek} as, for every $k\in \{1, \ldots, K\}$, $\eE_k = \menge{m \in \{1, \ldots, M\}}{k_m = k \text{ and } \Wc_m = \emp}$ (i.e., no overlap, no communication).
For completeness, we also introduce the notation $\overline{\Wc}_m = k_m \cup \Wc_m$.

For every $m \in \{1, \ldots, M\} \setminus (\bigcup_{k=1}^K \eE_k)$, workers $k' \in \Wc_m$ will need to send some vertex values from $\xb_{k'}$ to the worker $k_m$ so that it can compute $v_m$, as described in \eqref{eq:vm-Dmn-xn}. 
For every $k\in \{1, \ldots, K\}$, we denote by $\Rc_k  \subset \{1, \ldots, K\}\setminus \{k\}$ the set of workers $k'$ from which worker $k$ receives vertices. These workers store vertices belonging to a hyperedge of index $m\in \{1, \ldots, M\}$ such that $k_m = k$, but that are not stored on worker $k$. 
Similarly, we denote by $\Sc_k \subset \{1, \ldots, K\}\setminus \{k\}$ the set of workers $k'$ to which worker $k$ is sending vertices, i.e., all the workers $k'$ such that $k\in \Rc_{k'}$.

Communications will occur through hyperedges connecting different workers.  
For every $k'\in \Rc_k$, let 
\begin{equation*}
\eE_{(k,k')} = \menge{m\in \{1, \ldots, M\}}{k_m = k \text{ and } k' \in \Wc_m }, 
\end{equation*}
be the set of hyperedges containing the vertices from worker $k'$ required by worker $k$ to compute $v_m$, so that $\Rc_k = \menge{k' \in \{1, \ldots, K\} \setminus \{k\}}{\exists m \in \eE_{(k,k')} }$.
Thus, the set of all the hyperedges that will carry out some communication to worker $k$ is  
\begin{equation*}
\eE_{\Rc_k} = \bigcup_{k'\in \Rc_k} \eE_{(k,k')} = \menge{m \in \{1, \ldots, M\}}{k_m = k \text{ and } \Wc_m \neq \emp}.
\end{equation*} 
As a result, $\eE_k$ and $\eE_{\Rc_k}$ form a partition of $\overline{\eE}_k$, 
where $\eE_k$ corresponds to the set of weights $v_m$ that can be computed locally on worker $k$ (i.e. hyperedges inducing no communication), and $\eE_{\Rc_k}$ corresponds to the weights $v_m$ that necessitate vertices provided by other workers, and therefore communications.

Each hyperedge $e_m$ potentially sends vertices towards the corresponding worker $k_m$. To identify vertices that are communicated between workers, let 
\begin{equation*}
\eV_{(k,k')} = \menge{n \in \eV_{k'}}{\exists m \in \eE_{(k,k')} \text{ such that } n \in e_m} 
\end{equation*}  
be the indices of vertices that are received by worker $k$ from $k'$. 
The set of vertices sent to worker $k_m$ by all workers in $\Wc_m$ is denoted by $\eV_{\Wc_m} = \bigcup_{k' \in \Wc_m} \eV_{(k_m, k')}$. For completeness, we also introduce $\eV_{\overline{\Wc}_m} = \eV_{k_m} \cup \eV_{\Wc_m}$, the set of all the vertex indices necessary to perform computations associated with the $m$-th hyperedge.

For every $\xb \in \HHH$, we denote by $\xb_{(k,k')} = (x_n)_{n \in \eV_{(k,k')}}$ the vertices communicated from worker $k'$ to worker $k$ through some hyperedge $e_m$ such that $k_m=k$. 
For every $m\in \{1, \ldots,M\} \setminus (\bigcup_{k=1}^K \eE_k)$, let $\xb_{\overline{\Wc}_m} = (x_n)_{n \in \eV_{\overline{\Wc}_m}}$ be the concatenation of vertices stored on worker $k_m$ and vertices sent from all workers $k'\in \Wc_m$ to worker $k_m$ (i.e., all the vertices such that $D_{m,n} \neq 0$).

\begin{remark}
Let $k \in \{1, \ldots, K\}$. 
For every $k'\in {\Sc}_k$, since $\eV_{(k,k')}$ denotes the set of vertices that need to be communicated from worker $k'$ to worker $k$, the (reciprocal) set $\eV_{(k',k)} \subset \eV_k$ corresponds to the set of vertices that will be sent back from worker $k$ to $k'$ in the distributed implementation of the adjoint operator $\Db^*$. 
Similarly, with $\eE_{(k,k')}$ the set of hyperedges containing vertices that will be communicated from worker $k'$ to worker $k$, the set $\eE_{(k',k)} \subset \{1, \ldots, M\}$ corresponds to the hyperedge indices containing vertices that will be sent back from worker $k$ to worker $k'$.
\end{remark}

\subsubsection{Splitting of $\Db$ over the workers} 
\label{Sssec:splitting_D}

Eventually, we will introduce some notation to split the matrix $\Db$ over the workers.
%
In practice, the input space $\GGG$ and the output space $\HHH$ are partitioned using the sets $(\overline{\eE}_k)_{k\in \{1, \ldots, K\}}$ and $(\eV_k)_{k\in \{1, \ldots, K\}}$, respectively. For every $k \in \{1, \ldots, K\}$, let $\GGG_k = \cart_{m \in \overline{\eE}_k} \GG_m$ and $\HHH_k=\cart_{n \in \eV_k} \HH_n$ such that $\GGG = \cart_{1\leq k \leq K} \GGG_k$ and $\HHH = \cart_{1\leq k \leq K} \HHH_k$.

For every $k\in \{1, \ldots, K\}$, the subparts $\vb_k$ of $\Db\xb$ that belong to $ \GGG_k$ are stored on worker~$k$. The subparts of $\Db$ acting on $\HHH_k$ only are associated with vertices in $\eV_k$. For every $m \in \eE_k$, they will be denoted by $\Db_{m,k} = (D_{m,n})_{n \in \eV_k}$. These subparts of $\Db$ are involved in purely local computations only.
Similarly, for every $k' \in \Rc_k$, for every $m \in \eE_{(k,k')}$, the subparts of $\Db$ acting on vertices $\eV_{(k,k')}$ will be denoted by $\Db_{m,(k,k')} = (D_{m,n})_{n \in \eV_k \cup \eV_{(k,k')}} $.
In addition, for every $m \in \eE_{\Rc_k}$, we will denote by $\Db_{m,\overline{\Wc}_m} = (D_{m,n})_{n \in \eV_{\overline{\Wc}_m}}$ the subparts of $\Db$ acting on vertices stored either on worker $k_m$ or sent to worker $k_m$ by all other workers $k'\in \Wc_m$. Hence
\begin{equation}
    \label{eq:vm-formula}
    v_m = 
    \begin{cases}
    \Db_{m,k} \xb_k, &  \text{if } m \in \eE_k, \\
    \Db_{m,\overline{\Wc}_m} \xb_{\overline{\Wc}_m}, &   \text{otherwise}.
    \end{cases}
\end{equation}
%
As a result, the subparts of $\Db$ which must be stored on worker $k$ are the $\Db_{m,k}$, for every $m\in\eE_k$, and the $\Db_{m,\overline{\Wc}_m}$, for every $m \in \eE_{\Rc_k}$.

The notation given above is summarized in~\Cref{table:notation-hypergraph} and \Cref{table:notation2-hypergraph}. A simple example for block-sparse matrices is also provided below, and illustrated in~\Cref{fig:opsplitting}.

\begin{figure}
\begin{center}
\includegraphics[width=7cm]{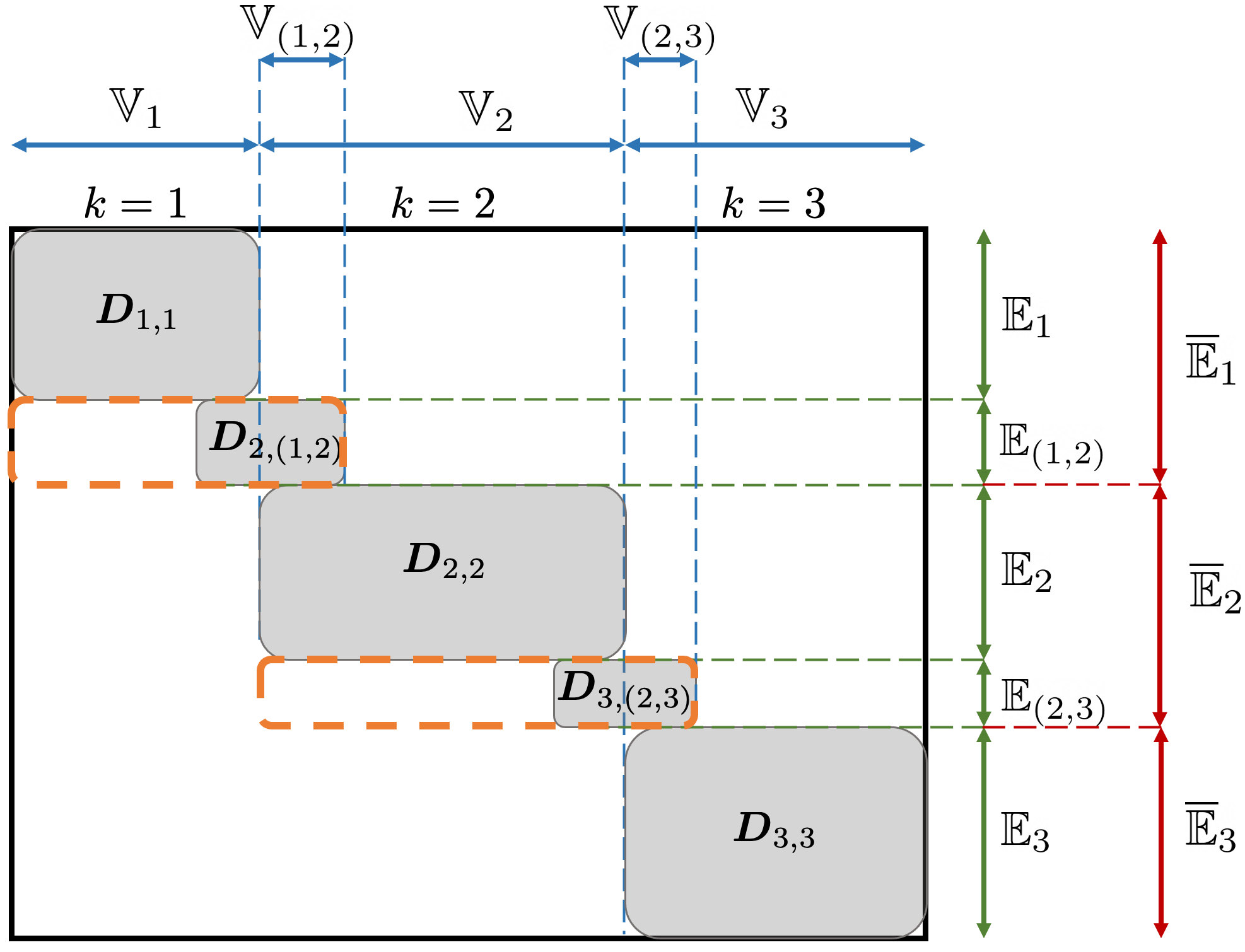}
\end{center}

\caption{\label{fig:opsplitting}\small
Distribution of vertices (horizontal lines) and hyperedges (vertical lines) of the hypergraph over $K=3$ workers, for a block-sparse matrix $\Db$ corresponding to~\Cref{ex:block-sparse}. 
}
\end{figure}

\begin{example}\label{ex:block-sparse}
To illustrate some of the notation introduced above, consider the block-sparse matrix $\Db$ shown in~\Cref{fig:opsplitting}, with $K=3$ workers. 
For the vertices, $\eV_1$ (resp. $\eV_2$ and $\eV_3$) contains the vertices handled on worker $k=1$ (resp. $k=2$ and $k=3$).
In addition, $\eV_{(1, 2)}$ (resp. $\eV_{(2, 3)}$) identifies the vertices of $\xb_2$ (resp. $\xb_3$) that will need to be communicated to worker $k=1$ (resp. $k=2$). 
For the hyperedges, $\overline{\eE}_1$ (resp. $\overline{\eE}_2$ and $\overline{\eE}_3$) identifies the hyperedge weights stored on worker $k=1$ (resp. $k=2$ and $k=3$). 
The set $\eE_1$ (resp. $\eE_2$ and $\eE_3$) identifies the hyperedges fully contained in worker $k=1$ (resp. $k=2$ and $k=3$).
Hyperedges in $\eE_{(1, 2)}$ (resp. $\eE_{(2, 3)}$) contain the vertices that will be communicated from worker $k=2$ to worker $k=1$ (resp. $k=2$ to $k=3$). 
In this example, only the blocks $\Big( (\Db_{m,k})_{m \in \eE_k} \Big)_{1 \le k \le 3}$, $(\Db_{m, (1,2)})_{m \in \eE_{(1,2)}}$ and $(\Db_{m, (2,3)})_{m \in \eE_{(2,3)}}$ are non-zero (e.g., convolution operator). 
Blocks $\Big( (\Db_{m,1})_{m\in \eE_1}, (\Db_{m, (1,2)})_{m\in \eE_{(1,2)}} \Big)$ 
(resp. $\Big( (\Db_{m,2})_{m \in \eE_2}, (\Db_{m, (2,3)})_{m \in \eE_{}(2, 3)} \Big)$ 
and $(\Db_{m,3})_{m \in \eE_3}$) are stored on worker $k=1$ (resp. $k=2$ and $k=3$). 
For worker $k=1$ (resp. $k=2$ and $k=3$), workers sending vertex values to this worker is given by $\Rc_1 = \{2\}$ (resp. $\Rc_2 = \{3\}$ and $\Rc_3 = \emp$). Similarly, workers receiving vertex values from this worker is given by $\Sc_1 = \emp$ (resp. $\Sc_2 = \{1\}$ and $\Sc_3 = \{2\})$.
For every $m\in \overline{\eE}_1$ (resp. $\overline{\eE}_2$ and $\overline{\eE}_3$), $k_m = 1$ (resp. $k_m = 2$ and $k_m = 3$). 
For every $m \in \eE_{(1, 2)}$, $\Wc_m = \Rc_1 = \Sc_2 = \{2\}$, and for every $m\in \eE_{(2,3)}$, $\Wc_m = \Rc_2 = \{3\}$. 
Finally, for every $m\in \eE_{\Rc_1}= \eE_{(1,2)}$ (resp. $ \eE_{\Rc_2}= \eE_{(2,3)}$), $\Db_{m, \overline{\Wc}_m} = (D_{m,n})_{n \in \eV_1 \cup \eV_{(1,2)}}$ (resp.  $\Db_{m, \overline{\Wc}_m} = (D_{m,n})_{n \in \eV_2 \cup \eV_{(2,3)}}$). This last notation corresponds to the orange rectangles in~\Cref{fig:blocksparse}. In this example, $\eE_{\Rc_3} = \emp$.
\end{example}

\subsection{Separability assumptions}
\label{Ssec:Model:separability}

We can now give the assumptions that will be used to design an efficient distributed SGS algorithm. 
Precisely, we will consider distribution~\eqref{eq:dist-SPA}, where the functions are assumed to satisfy the following separability conditions. 
\begin{assumption}  \label{ass:gen}\
\begin{enumerate}
    \item \label{ass:gen:ii}
    $h \colon \HHH \to \RR$ and $f  \colon \HHH \to \RX$ are additively separable on the workers, i.e. 
    \begin{equation}
        (\forall \xb \in \HHH) \quad
            h(\xb) = \sum_{k=1}^K h_k(\xb_k),
            \text{ and }
            f(\xb) = \sum_{k=1}^K f_k(\xb_k),
    \end{equation}
    where, for every $k \in \{1, \ldots, K\}$, $h_k \colon \HHH_k \to \RR$ and $f_k \colon \HHH_k \to \RX$.
    
    \item\label{ass:gen:iii}
    $g \colon \GGG \to \RX$ is additively separable, i.e.
    \begin{equation} 
        (\forall \vb  \in \GGG) \quad
            g(\vb) = \sum_{m=1}^M g_m(v_m), \label{eq:g-add}
    \end{equation}
    with, for every $m\in \{1, \ldots,M\}$, $g_m \colon \GG_m \to \RX$. 
\end{enumerate}
\end{assumption}
Further note that~\eqref{eq:phi} and~\eqref{eq:psi} can be written
\begin{align}
    (\forall (\vb, \ub) \in \GGG^2)\quad
    &  \phi_\alpha(\vb, \ub) = \sum_{m=1}^M \phi_{m, \alpha} (v_m, u_m) = \sum_{m=1}^M  \frac{1}{2\alpha^2} \| v_m - u_m \|^2, \label{eq:phi-add} \\
    &   \psi_\beta(\ub) = \sum_{m=1}^M \psi_{m, \beta} (u_m) = \sum_{m=1}^M  \frac{1}{2\beta^2} \| u_m \|^2 .  \label{eq:psi-add}
\end{align}
The above assumptions can be rewritten to highlight the separability over the workers. Indeed, there exists a permutation $\varrho \colon \GGG \to \GGG$ such that
\begin{align} 
    (\forall \xb \in \HHH) \quad
    \Db \xb
    = \left( \left( D_{m,n} \right)_{1 \le n \le N} \xb \right)_{1 \le m \le M}
    &= \varrho \left(\left( \begin{matrix}
    \left(  \Db_{m,k}\right)_{m \in \EE_k}  \xb_k     \\[0.1cm]
    \left(  \Db_{m,\overline{\Wc}_m}  \xb_{\overline{\Wc}_m} \right)_{m \in \eE_{\Rc_k}} 
    \end{matrix} \right)_{1 \le k \le K} \right). \label{eq:def:perm}
\end{align}
Reorganizing the vertex indices emphasizes the distinction between local computations and those for which communications are required.
For every $(\ub, \vb) \in \GGG^2$ and $\xb \in \HHH$, \eqref{eq:g-add}, \eqref{eq:phi-add} and \eqref{eq:psi-add} can thus be rewritten as 
\begin{align}
    g(\vb)
    &=  \sum_{k=1}^K \Big( 
        \sum_{m \in \overline{\eE}_k} g_m \left( v_m  \right) 
        \Big),    \label{eq:g-dist-simple}   \\
    \phi_\alpha(\Db \xb, \ub)
    &=  \sum_{k=1}^K \Big( 
        \sum_{m \in \overline{\eE}_k} \phi_{m, \alpha} \left( [\Db \xb]_m, u_m  \right) 
        \Big),  \\
    &=  \sum_{k=1}^K \Big( 
        \sum_{m \in \eE_k} \phi_{m,\alpha} \left( \Db_{m,k} \xb_k, u_m \right) 
        + \sum_{m \in \eE_{\Rc_k}}  \phi_{m,\alpha} \left( \Db_{m,\overline{\Wc}_m} \xb_{\overline{\Wc}_m}, u_m \right) 
        \Big),   \label{eq:phi-dist-simple}    \\
    \psi_\beta(\ub)
    &=  \sum_{k=1}^K \Big( 
        \sum_{m \in \overline{\eE}_k} \psi_{m,\beta} \left( u_m \right) 
        \Big),    \label{eq:psi-dist-simple} 
\end{align}
where $[\cdot]_m$ denotes the $m$-th element of its argument.

\section{Distributed Split Gibbs sampler}
\label{Sec:dist-single-algo}
\subsection{Proposed distributed algorithm}
\label{Ssec:algo}

Using the model from the previous section, we are now able to introduce a distributed block-coordinate version of Algorithm~\ref{algo:SPA-simple-sampling}. 
The term \emph{block-coordinate} is to be interpreted as in the optimization literature, where the variables from \eqref{eq:dist-SPA} are divided into $K$ blocks distributed over the $K$ workers, and updated in parallel.
The proposed distributed algorithm is the following.

\begin{proposition} \label{prop:synch-SPA-simple}
Consider a distribution~\eqref{eq:dist-SPA} satisfying \Cref{ass:psgla-spa,ass:gen}. Assume that the operator $\Db$ follows \Cref{mod:hypergraph}, and is split over workers $\{1, \ldots, K\}$ such that, for every $k\in \{1, \ldots, K\}$, $(\Db_{m,k})_{m \in \overline{\eE}_k}$ is stored on worker $k$. 
For every $k\in \{1, \ldots, K\}$, let $\xb_k^{(0)} \in \HHH_k$, $ \zb_k^{(0)} \in \GGG_k$, and $ \ub_k^{(0)} \in \GGG_k$. 
Let $(\xb^{(t)}, \zb^{(t)}, \ub^{(t)})_{1 \le t \le T}$ be samples generated by 
Algorithm~\ref{algo:single-synch},
where $\gamma \in ]0, (\lambda_h + \| \Db \|^2/\alpha^2)^{-1}[$ and, for every $k\in \{1, \ldots, K\}$, $(\wb^{(t)}_k)_{1 \le t \le T}$ is a sequence of i.i.d. standard Gaussian random variables in $\HHH_k$.
Let $\pi_{k, \alpha}$ denote the distribution
\begin{equation}
       \distc{k, \alpha}{\zb_k}{\vb_k, \ub_k}   
        \propto \exp \Bigg( - \sum_{m\in \overline{\eE}_k} \left( g_m(z_m) + \phi_{m,\alpha}( v_m  , z_m - u_m ) \right) \Bigg)
\end{equation}
where, for every $\xb \in \HHH$, $\vb_k = (v_m)_{m \in \overline{\eE}_k} = \varrho\left(\begin{matrix}
( \Db_{m,k} \xb_k )_{m \in \eE_k}  \\
( \Db_{m,\overline{\Wc}_m} \xb_{\overline{\Wc}_m} )_{m \in \eE_{\Rc_k}}
\end{matrix}\right)$. \\
Then, Algorithm~\ref{algo:single-synch} is equivalent to Algorithm~\ref{algo:SPA-simple-sampling}.
\end{proposition}

\begin{algorithm}[htbp]
    \small

    \For{$k = 1$ \KwTo $K$}{

        \For{$ k' \in \Sc_k$}{
            Send $(x_n^{(0)})_{n \in \eV_{(k',k)}}$ to worker $k'$;
        }

        \For{$k' \in \Rc_k$}{
            Receive $(x_n^{(0)})_{n \in \eV_{(k,k')}}$ from worker $k'$;
        }

        $\displaystyle \vb_k^{(0)} = \varrho
            \left(\begin{matrix}
            ( \Db_{m,k} \xb_k^{(0)} )_{m \in \eE_k}  \\
            ( \Db_{m,\overline{\Wc}_m} \xb_{\overline{\Wc}_m}^{(0)} )_{m \in \eE_{\Rc_k}}
            \end{matrix}\right)$; \\[0.2cm]
    }

    \For{$t = 0$ \KwTo $T$}{
        \For{$k= 1$ \KwTo $K$}{

            $(d_m^{(t)})_{m \in \overline{\eE}_k} 
                = \left( \phi_{m,\alpha}' (\cdot, z_m^{(t)}-u_m^{(t)}) (v_m^{(t)}) \right)_{m \in \overline{\eE}_k}$; \\[0.2cm]

            \For{$k' \in \Rc_k$}{
                $\displaystyle \widetilde{\db}_{(k',k)}^{(t)} = \sum_{m \in \eE_{(k',k)}} \Db_{m,k'}^* d_m^{(t)}$; \\
                Send $\widetilde{\db}_{(k',k)}^{(t)}$ to worker $k'$; \\[0.2cm]
            }

            \For{$k' \in \Sc_k$}{
                Receive $\widetilde{\db}_{(k,k')}^{(t)}$ from worker $k'$;
            }

            $\displaystyle \deltab_k^{(t)} = \sum_{m \in \overline{\eE}_k} \Db_{m,k}^* d_m^{(t)} + \sum_{k'\in \Sc_k} \widetilde{\db}_{(k,k')}^{(t)}$; \\[0.1cm]
            $\displaystyle \xb_k^{(t+1)} = \prox_{\gamma f_k}\left( \xb_k^{(t)} - \gamma \nabla h_k(\xb_k^{(t)}) - \gamma \deltab_k^{(t)} + \sqrt{2\gamma} \wb_k^{(t)} \right)$; \\[0.2cm]

            \For{$k' \in \Sc_k$}{
                Send $(x_n^{(t+1)})_{n \in \eV_{(k',k)}}$ to worker $k'$;
            }

            \For{$k' \in \Rc_k$}{
                Receive $(x_n^{(t+1)})_{n \in \eV_{(k,k')}}$ from worker $k'$; \\[0.2cm]
            }

            $\displaystyle \vb_k^{(t+1)} = \varrho
            \left(\begin{matrix}
            ( \Db_{m,k} \xb_k^{(t+1)} )_{m \in \eE_k}  \\
            ( \Db_{m,\overline{\Wc}_m} \xb_{\overline{\Wc}_m}^{(t+1)} )_{m \in \eE_{\Rc_k}}
            \end{matrix}\right)$; \\[0.3cm]
            $\zb_k^{(t+1)} \sim \distc{k,\alpha}{ \zb_k}{\vb_k^{(t+1)}, \ub_k^{(t)}}$; \\[0.2cm]
            $\ub_{k}^{(t+1)} \sim \Nc \Big( \frac{\beta^2}{\alpha^2+\beta^2} (\zb_{k}^{(t+1)} - \vb_{k}^{(t+1)}), \frac{\alpha^2+\beta^2}{\alpha^2 \beta^2} \Idb \Big)$;
        }
    }

    \caption{\small Proposed distributed SGS (distributed version of Algorithm~\ref{algo:SPA-simple-sampling}).}
    \label{algo:single-synch}
\end{algorithm}


All the computations described in Algorithm~\ref{algo:single-synch} are conducted simultaneously on each worker, after the necessary communication phases.

Before giving the proof of this result, we want to emphasize that communications to apply the operator $\Db$ and its adjoint $\Db^*$ are symmetric: the same workers are involved in communications, but the direction of the communications (send/receive) is reversed; the indices of the subparts that are communicated are not the same. Communications necessary to computations involving $\Db$ are received by worker $k$ from those of ${\cal R}_k$; communications necessary to computations involving $\Db^*$ are received by worker $k$ from those of ${\cal S}_k$ (see \Cref{ex:comp-comm-D}).

\begin{proof}
\Cref{ass:psgla-spa} ensures that \Cref{Prop:cv-SPA-simple} is verified (see~\eqref{eq:phi-psi-def}), and that the SGS algorithm can be instantiated as Algorithm~\ref{algo:SPA-simple-sampling}.

Using notation from \Cref{Sec:Model}, \eqref{eq:g-add}--\eqref{eq:psi-add} in \Cref{ass:gen} can be directly re-written as in~\eqref{eq:g-dist-simple}--\eqref{eq:psi-dist-simple}. Then, fixing $\vb = \Db\xb$, the conditional distributions~\eqref{eq:dist-SPA-cond:z}-\eqref{eq:dist-SPA-cond:u} can be rewritten as
\begin{align}
    \distc{\alpha}{\zb}{\vb, \ub}
    &   = \prod_{k=1}^K  \distc{k, \alpha}{\zb_k}{\vb_k, \ub_k}, \\
    \distc{(\alpha, \beta)}{\ub}{\vb, \zb} 
    &= \prod_{k=1}^K   \distc{k, (\alpha, \beta)}{\ub_k}{\vb_k, \zb_k},
\end{align}
respectively, where, for every $k\in \{1, \ldots, K\}$,
\begin{align}
    \distc{k, \alpha}{\zb_k}{\vb_k, \ub_k}
    &   \propto \exp \left( - \sum_{m\in \overline{\eE}_k} \left( g_m(z_m) + \phi_{m,\alpha}(v_m  , z_m - u_m ) \right)\right) ,  \\
    \distc{k, (\alpha, \beta)}{\ub_k}{\vb_k, \zb_k}
    &   \propto \exp \left( - \sum_{m\in \overline{\eE}_k} \left( \phi_{m, \alpha}( v_m, z_m - u_m) - \psi_{m, \beta} (u_m) \right) \right),
    \\
    &= \Nc \bigg( \frac{\beta^2}{\alpha^2+\beta^2} (\zb_k - \vb_k), \frac{\alpha^2+\beta^2}{\alpha^2 \beta^2} \Idb \bigg), \nonumber
\end{align}
for $\vb_k = (v_m)_{m \in \overline{\eE}_k} = ([\Db \xb]_m)_{m \in \overline{\eE}_k}$. 
Hence, for every iteration $t \in \{0, \ldots, T\}$ of the Algorithm~\ref{algo:SPA-simple}, and for every worker $k \in \{1, \ldots, K\}$,
the sampling of $\zb_k^{(t+1)}$ and $\ub_k^{(t+1)}$ only requires partial information from $\big( \vb_k^{(t+1)}, \ub_k^{(t)} \big) $ and $\big( \vb_k^{(t+1)}, \zb_k^{(t+1)} \big) $, respectively. 
Algorithm~\ref{algo:SPA-simple-sampling} is thus equivalent to 
\begin{equation}    \label{prop:synch-SPA-simple:proof:algo1}
\begin{array}{l}
    \text{for } t = 0,1, \ldots, T	\\
    \left\lfloor
    \begin{array}{l}
        \xb^{(t+1)} = \prox_{\gamma f} \Big( \xb^{(t)} - \gamma \nabla h(\xb^{(t)}) - \gamma \Db^* \Big( \nabla \phi_\alpha(\cdot , \zb^{(t)}-\ub^{(t)})((v_m^{(t)})_{m \in \overline{\eE}_k}) \Big) \\
        \qquad \qquad \qquad \quad + \sqrt{2\gamma} \, \wb^{(t)} \Big) \\
        \text{for } k = 1, \ldots, K	\\
        \left\lfloor
        \begin{array}{l}
            \vb_k^{(t+1)}
            = ([\Db \xb^{(t+1)}]_m)_{m \in \overline{\eE}_k}, \\
            \zb_k^{(t+1)} 
                \sim \distc{k, \alpha}{\zb_k}{ \vb_k^{(t+1)}, \ub_k^{(t)}}, \\
            \ub_k^{(t+1)} 
                \sim \Nc \bigg( \frac{\beta^2}{\alpha^2+\beta^2} (\zb_k^{(t+1)} - \vb_k^{(t+1)}), \frac{\alpha^2+\beta^2}{\alpha^2 \beta^2} \Idb \bigg).
        \end{array}
        \right.
    \end{array}
    \right.
\end{array}
\end{equation}

It remains to show that, for every $t\in \{0, \ldots, T\}$, the computation of $\xb^{(t+1)}$ can be parallelized over $k\in \{1, \ldots, K\}$. 
Since $(\eV_k)_{1 \le k \le K}$ is a partition of $\{1, \ldots, N\}$, then, for every $k\in \{1, \ldots, K\}$, $\xb_k^{(t+1)}$ will be given by the $k$-th vertex values of the random variable generated from the PSGLA transition kernel.

The gradient of \eqref{eq:phi-dist-simple} with respect to $\xb$ is required to compute $(\xb_k^{(t+1)})_{1 \le k \le K}$.
Using the chain rule on $\xb \in \HHH \mapsto (\phi_\alpha( \cdot, \zb-\ub) \circ \Db) (\xb)$, for $(\zb, \ub)\in \GGG^2$ fixed, we obtain, for every $\xb \in \HHH$,
\begin{equation}
    \nabla_{\xb} \phi_\alpha(\Db \xb , \zb-\ub)
    = \nabla_{\xb} (\phi_\alpha( \cdot, \zb-\ub) \circ \Db) (\xb) 
    = \Db^* \nabla \phi_\alpha (\cdot, \zb-\ub) (\Db \xb).
\end{equation}
According to \eqref{eq:phi-dist-simple}, for every $(\vb, \zb) \in \GGG^2$, we have
\begin{align}
    \nabla_{\vb} \phi_\alpha (\vb, \zb)
    = \varrho \left( \left( \begin{matrix} 
    \Big( \phi'_{m,\alpha}(v_m, z_m) \Big)_{m \in \eE_k} \\
    \Big( \phi'_{m, \alpha}(v_m, z_m) \Big)_{m \in \eE_{\Rc_k}}
    \end{matrix} \right)_{1 \le k \le K} \right),
\end{align}
where $\varrho$ is the permutation operator defined in~\eqref{eq:def:perm}.
For every $k \in \{1, \ldots, K\}$, every $\xb \in \HHH$ and $(\zb, \ub) \in \GGG^2$, let $(d_m)_{m\in \overline{\eE}_k}$ be defined as
\begin{align}
    d_m 
    := \phi'_{m,\alpha}(\cdot , z_m-u_m) (v_m)
    &= \begin{cases}
    \phi'_{m,\alpha}(\cdot , z_m-u_m) (\Db_{m,k} \xb_k),
    &   \text{if } m \in {\eE}_k,  \\
    \phi'_{m,\alpha}(\cdot , z_m-u_m) (\Db_{m,\overline{\Wc}_m} \xb_{\overline{\Wc}_m}),
    &   \text{if } m \in \eE_{\Rc_k}.
    \end{cases}     
\end{align}
Finally,
\begin{equation}
    \nabla_{\xb} \phi_\alpha(\Db \xb , \zb-\ub) = \Db^* \varrho \Bigg( \Big( (d_m)_{m \in \overline{\eE}_k} \Big)_{1 \le k \le K} \Bigg). \label{proof:gradient}
\end{equation}
To extract the $k$-th block from this gradient, \eqref{proof:gradient} can be decomposed as follows:
\begin{align}
    \deltab_k 
    &:= [\nabla_{\xb} \phi_\alpha(\Db \xb , \zb-\ub)]_k \nonumber \\
    &= \sum_{m \in \overline{\eE}_k} \Db_{m,k}^* d_m + \sum_{k' \in \Sc_k} \sum_{m \in \eE_{(k,k')}} \Db_{m, k}^* d_m.  \label{prop:synch-SPA-simple:sum-grad-glob}
\end{align}
Figure~\ref{fig:comp-comm-D} illustrates this for the simple example described in Figure~\ref{fig:opsplitting}. 

The second term in \eqref{prop:synch-SPA-simple:sum-grad-glob} gathers information received from the set of neighbours $\Sc_k$ of worker $k$ that is necessary to computations involving $\Db^*$.
Then, using the notation 
\begin{equation}
    (\forall k' \in \Sc_k)\quad
    \widetilde{\db}_{(k,k')} = \sum_{m \in \eE_{(k,k')}} \Db_{m,k}^* d_m,
\end{equation}
we obtain
\begin{equation} \label{eq:aggregate_for_adjoint}
    \deltab_k
    = \sum_{m \in \overline{\eE}_k} \Db_{m,k}^* d_m + \sum_{k' \in \Sc_k} \widetilde{\db}_{(k,k')} ,
\end{equation}
where, for every $k' \in \Sc_k$, $\widetilde{\db}_{(k,k')}$ is computed on worker $k'$ and communicated to worker $k$ to form $\deltab_k$.
Therefore the computations performed on a given worker $k \in \{1, \ldots, K\}$ include the computation of \textit{local} gradients $(\Db_{m,k}^* d_m)_{m \in \overline{\eE}_k}$ as well as the gradients to be sent to neighbours, i.e., $(\widetilde{\db}_{(k',k)})_{k' \in \Rc_k}$. 
Note that the role of the sets $\Rc_k$ and $\Sc_k$ is exchanged when the adjoint operator $\Db^*$ is considered.

Using the above notation, and using \cite[Prop. 24.11]{Bauschke2017CAMOTHS} on the additively separable functions $h$ and $f$, leads to the conclusion that Algorithm~\ref{algo:single-synch} is equivalent to \eqref{prop:synch-SPA-simple:proof:algo1}, and therefore to Algorithm~\ref{algo:SPA-simple-sampling}.
\end{proof}

\subsection{Illustration of operator computation}
\label{ex:comp-comm-D}

In this section, we give an example to explain the communications involved in Algorithm~\ref{algo:single-synch}. \Cref{fig:dsgs} describes one iteration of Algorithm~\ref{algo:single-synch}.
Figure~\ref{fig:comp-comm-D} describes the application of the forward $\Db$ and backward $\Db^*$ operators.

\Cref{fig:dsgs}~(a) illustrates one iteration of the proposed distributed sampling Algorithm~\ref{algo:single-synch} implemented on an SPMD architecture, with local computations conducted on a worker $k\in \{1, \ldots, K\}$. \Cref{ass:gen} ensures that most of the related operations can be conducted independently on each worker. 
Applying the forward operator $\Db$ and its adjoint $\Db^*$ requires communications illustrated in~\Cref{fig:dsgs}~(b).

\begin{figure}[t]
    \centering
    \small
    \captionsetup[subfigure]{justification=centering}
    \begin{subfigure}[b]{0.53\textwidth}
        \centering
        \includegraphics[keepaspectratio, height=0.23\textheight]{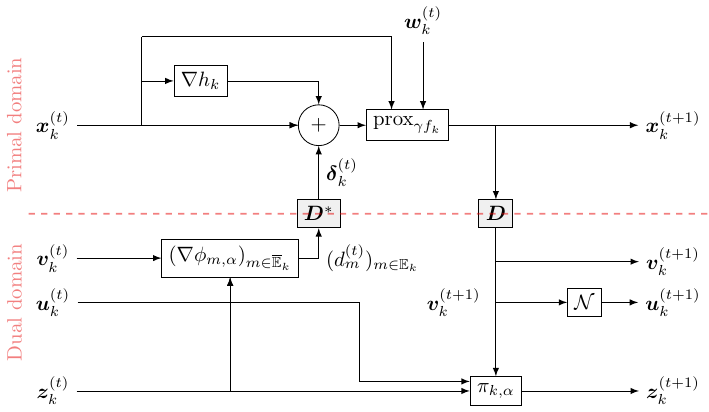} 
        \caption{Iteration $k$ of the proposed sampler}
        \label{fig:dsgs:iteration}
    \end{subfigure}
    \hfill
    \begin{subfigure}[b]{0.45\textwidth}
        \centering
        \includegraphics[keepaspectratio, height=0.23\textheight]{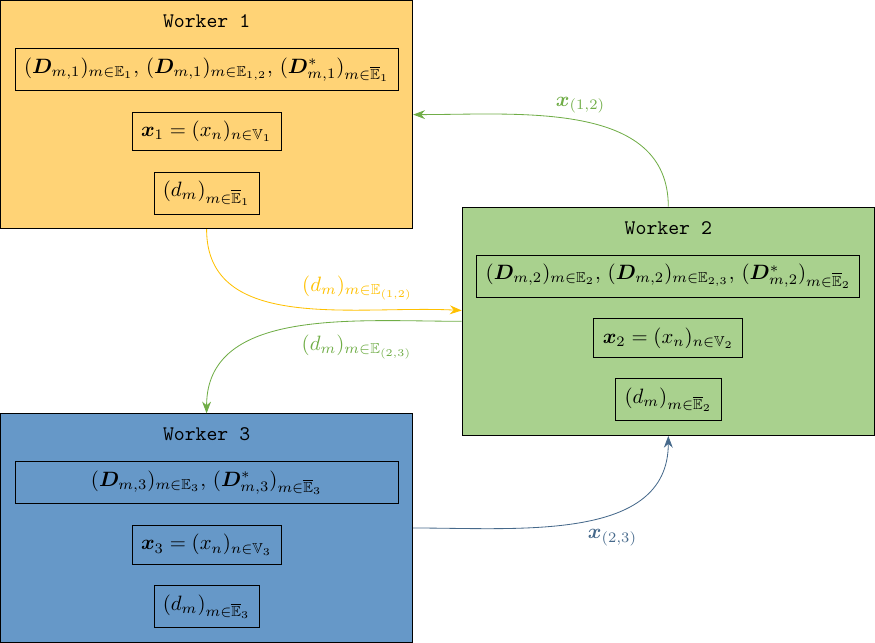}
        \caption{Communications between workers}
        \label{fig:dsgs:communication-graph}
    \end{subfigure}

    \caption{\small (a) Operations required by one iteration $k \in \eN^*$ of Algorithm~\ref{algo:single-synch}, a distributed version of Algorithm~\ref{algo:SPA-simple-sampling} using an SPMD architecture. Shaded rectangles correspond to operations for which communications are required. (b) Communications required by the distributed implementation of Algorithm~\ref{algo:single-synch}, associated with the operator $\Db$ illustrated in~\Cref{fig:comp-comm-D}. 
    Local variables are displayed in the boxes, and variables exchanged between workers are represented with colored arrows.}
    \label{fig:dsgs}
\end{figure}

\begin{figure}[ht]
\centering
\begin{tabular}{cc}\small
    \includegraphics[width=0.48\textwidth]{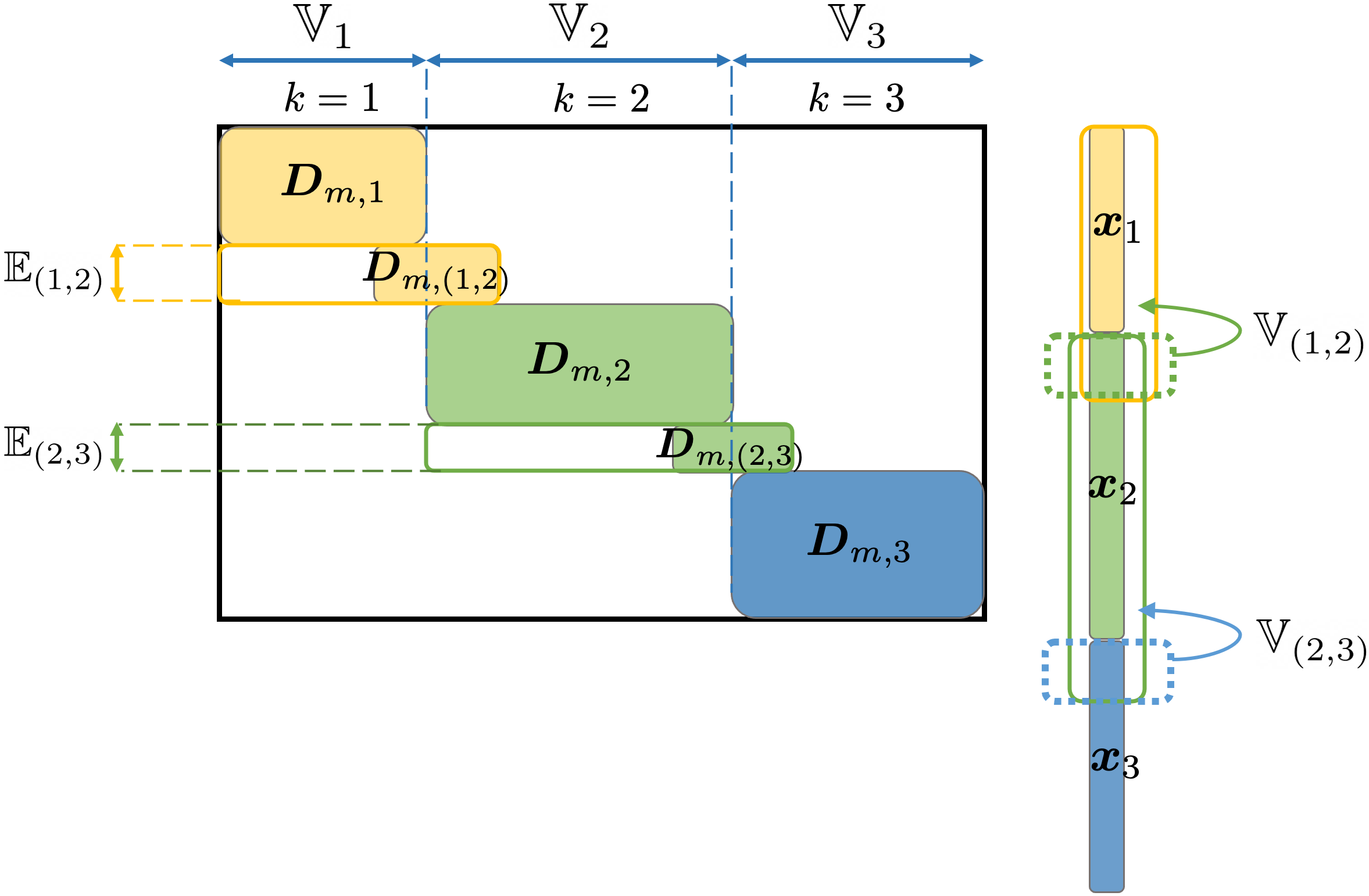} 
    &   \includegraphics[width=0.4\textwidth]{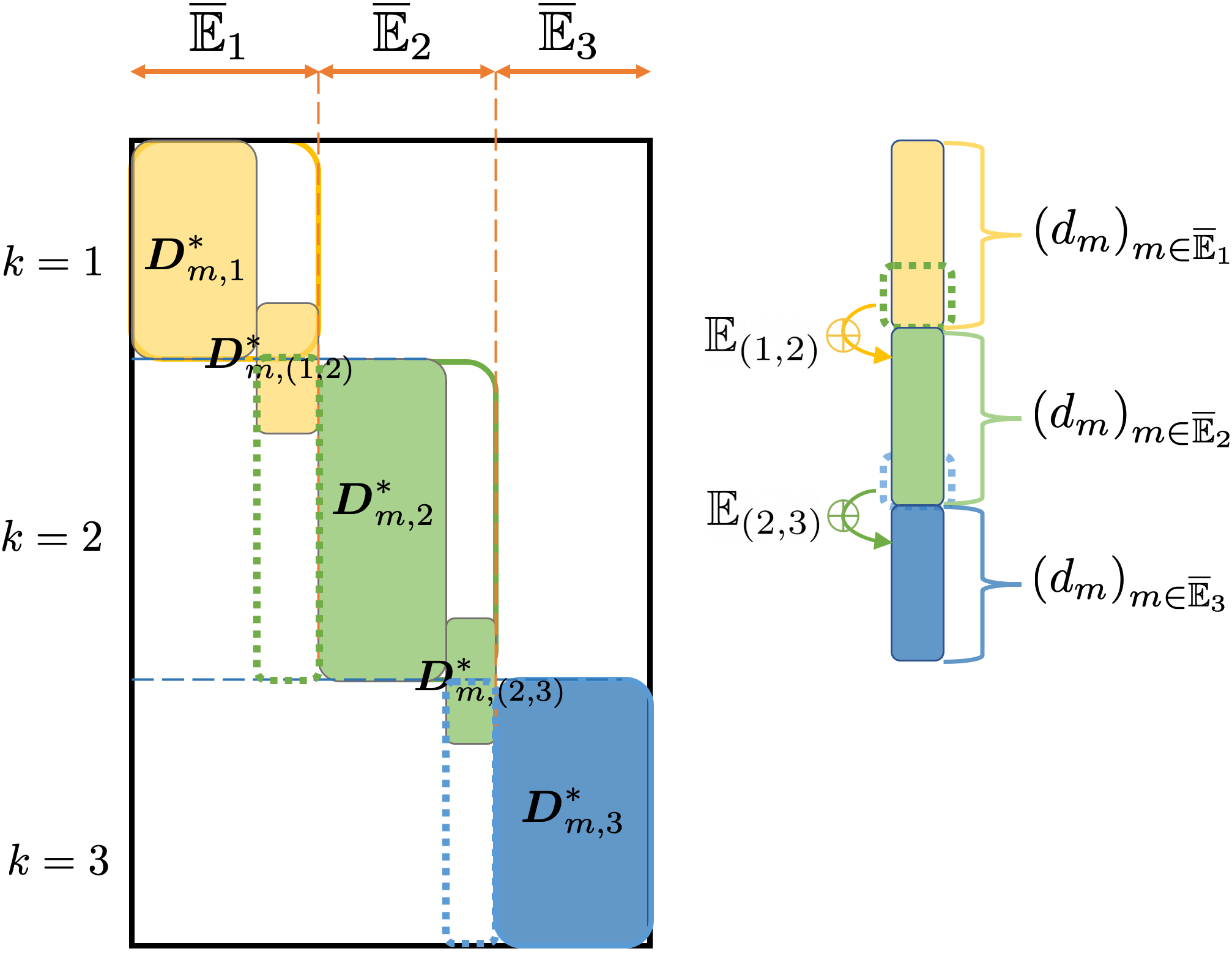}  \\
    (a) Computation of $\Db \xb$ 
    &   (b) Computation of $\Db^* \db$
\end{tabular}

\caption{\label{fig:comp-comm-D}\small
Illustration of the computation of (a) $\Db \xb$ for $\xb \in \HHH$, and (b) $\Db^* \db$ for $\db \in \GGG$. $\Db$ corresponds to the block-sparse matrix from Figure~\ref{fig:opsplitting}, and $\Db^*$ to its adjoint. The yellow, green and blue colours identify how $\Db$, $\Db^*$, $\xb$ and $\db$ are distributed over the workers $k=1$, $k=2$, and $k=3$, respectively. 
Communications are illustrated by coloured arrows over subparts of $\xb$ and $\db$, respectively. 
Continuous coloured lines on $\xb$ and $\db$ identify vertices and hyperedge weights required to perform local computations. The corresponding blocks in $\Db$ and $\Db^*$ are emphasized accordingly. Vertices and hyperedge weights to be sent from a worker to another are delineated in dashed lines, with the corresponding blocks within $\Db$ and $\Db^*$ highlighted accordingly.
(a) Vertices are on columns, and hyperedges on rows. 
(b) Vertices are on rows, and hyperedges on columns. The ``plus'' symbol emphasizes that hyperedge weights are aggregated upon reception.
}
\end{figure}

\Cref{fig:comp-comm-D}(a) illustrates the computation of $\Db \xb$, for $\xb \in \HHH$ with the block-sparse matrix $\Db$ of Example~\ref{ex:block-sparse}, see also Figure~\ref{fig:blocksparse}. Figure~\ref{fig:comp-comm-D}(b) illustrates the computation of $\Db^* \db$, for $\db \in \GGG$ with some adjoint symmetry. \\
Subparts $\eV_{(1,2)}$ of $\xb_2$ and $\eV_{(2,3)}$ of $\xb_3$ need to be communicated to workers $1$ and $2$, respectively, to compute $\Db \xb$. 
In Figure~\ref{fig:comp-comm-D}(a), for worker $k=1$, the quantity $(\Db_{m,\overline{\Wc}_m} \xb_{\overline{\Wc}_m})_{m\in \eE_{\Rc_1}}$ can be computed by multiplying the subpart of $\Db$ corresponding to $(\Db_{m, \overline{\Wc}_m})_{m \in \eE_{\Rc_1}}$ (see \Cref{fig:comp-comm-D}~(a), horizontal rectangle in continuous yellow lines) with $(\xb_{\overline{\Wc}_m})_{m\in \eE_{\Rc_1}} = (\xb_1, \xb_{(1,2)})$ (vertical rectangle in continuous yellow lines). 
This quantity is computed on worker $1$, once worker $2$ has communicated $\xb_{(1,2)}$ to worker $k=1$. 
Similarly, for worker $k=2$, the quantity $(\Db_{m,\overline{\Wc}_m} \xb_{\overline{\Wc}_m})_{m\in \eE_{\Rc_2}}$ can be computed by multiplying the subpart of $\Db$ corresponding to $(\Db_{m , \overline{\Wc}_m})_{m \in \eE_{\Rc_2}}$ (horizontal rectangle in continuous green lines) with $(\xb_{\overline{\Wc}_m})_{m\in \eE_{\Rc_2}} = (\xb_2, \xb_{(2,3)})$ (vertical rectangle in continuous green lines). 
This quantity is computed on worker $2$, once worker $3$ has communicated $\xb_{(2,3)}$ to worker $k=2$. \\
To compute $\Db^* \db$, communications are performed after computing subparts of $\Db^* \db$. 
In Figure~\ref{fig:comp-comm-D}(b), for worker $k=1$, the quantity $[\Db^* \db]_1$ can be computed by multiplying the subpart of $\Db^*$ corresponding to $(\Db^*_{m, 1})_{m \in \overline{\eE}_1}$ (continuous-line yellow rectangle, top left) with $(d_m)_{m \in \overline{\eE}_1}$. For the computation of $[\Db^* \db]_1$, no communication with other workers is needed in this example. 
For worker $k=2$, $[\Db^* \db]_2$ needs to be decomposed between parts of $\Db^*$ stored on worker $k=2$, and parts of $\Db$ that are stored on other workers, i.e., $k=1$ for this example. For the parts of $\Db^*$ stored on worker $k=2$, the subpart of $\Db^*$ corresponding to $(\Db^*_{m, 2})_{m \in \overline{\eE}_2}$ (continuous-line green rectangle) needs to be multiplied with $(d_m)_{m \in \overline{\eE}_2}$. For the parts of $\Db^*$ stored on $k=1$, the subpart of $\Db^*$ corresponding to $(\Db_{m, 2}^*)_{m \in \eE_{(1,2)}})$ (dotted-line green rectangle) needs to be multiplied with $(d_m)_{m\in \eE_{(1,2)}}$ (dotted-line green rectangle on vector $\db$). This second part is computed on worker $k=1$ (using only the yellow part of the dotted-line green rectangle), and then communicated and aggregated on worker $k=2$.  
For worker $k=3$, the decomposition of the quantity $[\Db^* \db]_3$ is similar to the one taken for worker $k=2$.

\subsection{Distributed SPMD architecture}
\label{Ssec:algo:dist-spmd}

The distributed block-coordinate sampler described in the previous section can benefit from an implementation on a Single Program Multiple Data (SPMD) architecture. 
In contrast with a client-server configuration, all the workers 
execute the same task on a subset of each block of parameters and observations~\cite{Darema2001}. This enables the hypergraph structure of $\Db$ to be exploited to reduce the number and volume of the communications.

In practice, \Cref{mod:hypergraph} and \Cref{ass:gen} ensure that most of the operations of Algorithm~\ref{algo:SPA-simple-sampling} are compatible with an SPMD architecture.
The separability Assumptions~\ref{ass:gen}-\ref{ass:gen:ii} and~\ref{ass:gen:iii} ensure that $K$ conditionally independent blocks 
can be formed for variables $\vb \in \GGG$ in the dual domain.
The separability of $g$ further implies that the evaluation of its proximity operator is easy to parallelize. Combining the definition of $\phi_{\alpha}$ and $\psi_{\beta}$ 
(\Cref{ass:psgla-spa}-\ref{ass:psgla-spa:phi_psi}) 
with the structure of $\Db$ finally enables Algorithm~\ref{algo:SPA-simple-sampling} to be reformulated using blocks of parameters, each stored on (only) one of the $K$ workers. Using an SPMD architecture for Algorithm~\ref{algo:single-synch} offers several advantages listed below.

\begin{enumerate}
    \item Load balancing: computation costs can be equally shared among the workers, as they all operate similar tasks on a subset of (overlapping) parameters.
    \item Parallelization flexibility: an SPMD architecture can be readily used to address~\eqref{eq:dist-gen-lin} under \Cref{ass:gen}, whereas a client-server architecture cannot (no conditionally independent blocks of variables as in~\cite{Rendell2021})
    \item Memory and computing costs per worker: each worker can be assigned 
    a block of parameters. This opportunity can significantly reduce the computing and memory costs per worker.
    \item Communications and data locality:
    most of the parameters required to perform operations on a worker can be directly stored on the same worker, and do not require to be communicated (\emph{data locality}). The conditions given in~\Cref{Sssec:Model:cond-eff-dist} guarantee that communication costs are limited: only a few elements need to be retrieved by each worker from a small number of connected workers.
\end{enumerate}

\if0\production
{\Cref{appendix:SPA-PSGLA-multi} } 
\else
{The online appendix} 
\fi
addresses the multi-term extension of~\eqref{eq:dist-gen-lin}. Note that a client-server architecture can also be used to address this case. However, it cannot take advantage of the structure of the hypergraphs underlying the linear operators $(\Db_i)_{1 \leq i \leq I}$. The number of workers it can accommodate is also restricted to the number of conditionally independent blocks of variables, as in~\cite{Vono_etal_2020, Rendell2021} (see \Cref{Ssec:dist:client-server-vs-spmd}).

\section{Application to supervised image deconvolution}
\label{Sec:application}
To show the performance of the proposed distributed SGS, we use it to solve an image deconvolution problem. 
Image deconvolution is an inverse problem that consists in inferring an unknown variable $\overline{\xb}$ from observations $\yb$. 
Observations and parameters are typically related by a model of the form
\begin{equation} \label{eq:full_model}
    \yb = \Dc(\Ab \overline{\xb}),
\end{equation}
where the linear operator $\Ab $ models the acquisition process, and $\Dc$ models random perturbations -- referred to as noise -- damaging the clean data $\Ab \overline{\xb}$. Bayesian inference relies on the posterior distribution of the random variable $\xb$ to estimate the true value $\overline{\xb}$. The posterior distribution, often of the form~\eqref{eq:dist-gen-lin}, combines information from the likelihood -- related to the observations $\yb$ -- and the prior. For instance, in image processing, a usual choice consists in promoting sparsity in a selected basis, e.g., a gradient basis leading to the total variation (TV) regularization~\cite{Rudin1992}, or a wavelet basis~\cite{Mallat2009}. Prior information can also encompass constraints based on the physics of the data acquisition process, such as nonnegativity for intensity images~\cite{Cai2018, Thouvenin2021}, or polarization constraints~\cite{Birdi2018}.

We consider a supervised image deconvolution problem corrupted by Poisson noise. The induced hypergraph structures are used to adopt an SPMD strategy. Note that another application of the proposed sampler has also been studied in \cite{Thouvenin2022eusipco} for an inpainting problem (\emph{i.e.}, with $\Ab$ a selection matrix) corrupted by additive white Gaussian noise under a TV prior.

The application example presented in this section is associated with a distribution involving $\ncomposite = 2$ composite terms. The notation used below corresponds to the one introduced in
\if0\production
{~\Cref{appendix:SPA-PSGLA-multi}}%
\else
{~the online appendix}%
\fi
for distributions with multiple composite terms.

\subsection{Problem statement}
\label{Ssec:application:problem_statement}

Supervised Poisson deconvolution aims at inferring an unknown image $\overline{\xb} \in \RR^N$ from observations $\yb = (y_m)_{1\le m \le M} \in \RR^M$ such that
\begin{equation} \label{eq:model_convolution}
    (\forall m \in \{1, \ldots, M\}) \quad
    y_m \sim \Pc ( [\Db_1 \overline{\xb}]_m ),
\end{equation}
where $\Db_1 \in \RR^{M \times N}$ is a convolution operator derived from a kernel of size $L = L_1 \times L_2 \ll N$, and $\Pc(\mu)$ is a Poisson distribution with mean $\mu$. 
In this context, $\overline{N}=N$ and, for every $n\in \{1, \ldots, N\}$, $\HH_n = \RR$. Similarly, $\overline{M}=M$ and, for every $m\in \{1, \ldots, M\}$, $\GG_m = \RR$.

We propose to solve this problem with 
a hybrid prior combining a non-negativity constraint and a TV regularization~\cite{Rudin1992}. Such prior has for instance been considered in \cite{Figueiredo2010,Vono2019icassp}. 
The resulting posterior distribution is given by
\begin{equation}
    \dist{}{\xb} \propto \exp \left( -f(\xb) - g_1(\Db_1 \xb) - g_2(\Db_2 \xb) \right),
\end{equation}
where
$g_1\colon \mathbb{R}^M \to ]-\infty, +\infty] \colon (z_m)_{1 \le m \le M}  \mapsto \sum_{m=1}^M g_{1,m}(z_m)$ is the data-fidelity term, with
\begin{align} 
    &(\forall m \in \{1, \ldots, M\}) \quad
    g_{1,m}(z_m) =  - y_m \log( z_m ) + z_m, \label{eq:data_fidelity}
\end{align}
due to the Poisson distribution, $f =\iota_{[0,+\infty[^N} \colon \RR^N \to ]-\infty, +\infty]$ is the indicator function of the positive orthant, and $g_2 \circ \Db_2$ models the discrete isotropic TV~\cite{Rudin1992}. 
Precisely, $\Db_2 \colon \RR^N \to \RR^{2\times N}$ is the concatenation of the vertical and horizontal discrete gradients, and $g_2: \mathbb{R}^{2 \times N} \to ]-\infty, +\infty]$ is the $\ell_{2,1}$-norm
\begin{equation}
    \begin{split}
        &(\forall \zb = (\zb_n)_{1\leq n \leq N} \in \RR^{2 \times N}) \quad g_2(\zb) = \sum_{n=1}^{N} g_{2, n} (\zb_n), \\
        &(\forall n \in \{1, \ldots, N\}) (\forall \zb_n \in \RR^2) \quad
        g_{2, n}(\zb_n) = \kappa \| \zb_n \|_2, \quad \text{where }\kappa > 0.
    \end{split}
\end{equation}

As a first step towards a distributed implementation, the AXDA approach (see%
\if0\production
{~\Cref{appendix:SPA-PSGLA-multi}}%
\else
{~the online appendix}%
\fi
for details) is used to approximate $ \dist{}{\xb} $ by
\begin{multline} \label{eq:posterior_deconvolution}
     \dist{\alphab, \betab}{\xb,(\zb_\idcomposite, \ub_\idcomposite)_{1 \le \idcomposite \le 2}} 
     \!\propto\! \exp \bigg( \!\! - f(\xb) - \sum_{\idcomposite=1}^2 \big( g_\idcomposite(\zb_\idcomposite) + \phi_{\idcomposite, \alpha_\idcomposite}(\Db_\idcomposite \xb, \zb_\idcomposite-\ub_\idcomposite)  + \psi_{\idcomposite, \beta_\idcomposite}(\ub_\idcomposite) \big) \bigg),
\end{multline}
where $(\alpha_1, \alpha_2, \beta_1, \beta_2) \in ]0, +\infty[^4$, and the functions $(\phi_{\idcomposite, \alpha_\idcomposite}, \psi_{\idcomposite, \beta_\idcomposite})_{1 \le \idcomposite \le 2}$ are defined in~\eqref{eq:phi-psi-def}.

The directed acyclic graph reported in~\Cref{fig:dag} summarizes the structure of the approximate posterior distribution~\eqref{eq:posterior_deconvolution}, highlighting dependencies between the variables. Note that $(\zb_1, \ub_1)$ and $(\zb_2, \ub_2)$ are the only conditionally independent blocks of variables.

\begin{figure}[htbp]
    \centering
    \includegraphics[keepaspectratio, width=0.45\textwidth]{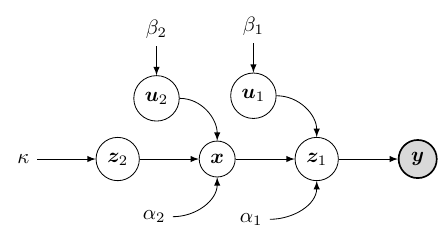}
    
    \vspace*{-0.5cm}
    
    \caption{\small
    Directed acyclic graph describing the factorization of the approximate posterior distribution~\eqref{eq:posterior_deconvolution}. Circled variables correspond to random variables. Other variables are fixed \emph{a priori}.} 
    \label{fig:dag}
\end{figure}

\subsection{Proposed SPMD implementation}
\label{Ssec:application:spmd}

The proposed approach is applicable when a single operator is involved in the model (as in Section~\ref{Sec:dist-single-algo}), but also when multiple operators are involved (see%
\if0\production
{~\Cref{appendix:SPA-PSGLA-multi}}%
\else
{~the online appendix}%
\fi
), as in the present example. 
Note that a client-server architecture could also be considered. However, it would drastically increase the communication costs, as full-size variables would need to be duplicated on all workers or exchanged (see \Cref{Ssec:algo:dist-spmd}).
Note that~\eqref{eq:model_convolution} indicates that the variables $(y_m)_{1 \leq m \leq M}$ are assumed independent. This implies that the observations can be partitioned into $K$ statistically independent blocks $\yb = (\yb_k)_{1 \leq k \leq K}$, with $\yb_k \in \GGG_k$ a block of observations to be stored on the worker $k$. In this case, each worker can be assigned a block of observations $\yb_k$ and a corresponding block of parameters $\xb_k$. Thus, no observation needs to be exchanged between the workers.

We propose to use the proposed distributed SGS%
\if0\production
{ Algorithm~\ref{algo:multi-synch}}%
\else
{ Algorithm~2 from the online appendix}%
\fi
with $\ncomposite=2$ linear operators. In addition, for every $t\in \{0, \ldots, T\}$ and $k\in \{1, \ldots, K\}$, we use PSGLA transitions to compute $(\zb_{1, k}^{(t+1)}, \zb_{2, k}^{(t+1)})$, while $(\ub_{1, k}^{(t+1)}, \ub_{2, k}^{(t+1)})$ are drawn from their conditional distribution. More precisely, we have, for $\idcomposite \in \{1, 2\}$,
\begin{align}
    \zb_{\idcomposite, k}^{(t+1)}
    &=   \prox_{\eta_\idcomposite, g_{\idcomposite,k}} \left( \zb_{\idcomposite,k}^{(t)} - \eta_\idcomposite\alpha_\idcomposite^{-2} \big( \zb_{\idcomposite,k}^{(t)} - \vb_{\idcomposite,k}^{(t+1)} + \ub_{\idcomposite,k}^{(t)} \big) + \sqrt{2\eta_\idcomposite} \xib_{\idcomposite,k}^{(t)} \right) ,
\end{align}
where $g_{\idcomposite, k} = \sum_{m\in \overline{\eE}_{\idcomposite,k}} g_{\idcomposite,m}$, $\eta_\idcomposite \in ]0, \alpha_\idcomposite^{-2}[$, and $\xi_{\idcomposite,k}^{(t)} \sim \Nc(0, \Idb)$.
The proximity operators involved in the resulting algorithm can be found, e.g., in~\cite{Figueiredo2010,Komodakis2015}.
The distributed implementation of $(\Db_\idcomposite)_{1 \leq \idcomposite \leq 2}$ and $(\Db_\idcomposite^*)_{1 \leq \idcomposite \leq 2}$ is detailed below, using a 2D Cartesian grid of $K$ workers.

\begin{figure}[htbp]
    \centering
    \begin{tabular}{@{}c@{}c@{}}
    \includegraphics[keepaspectratio, width=0.5\textwidth]{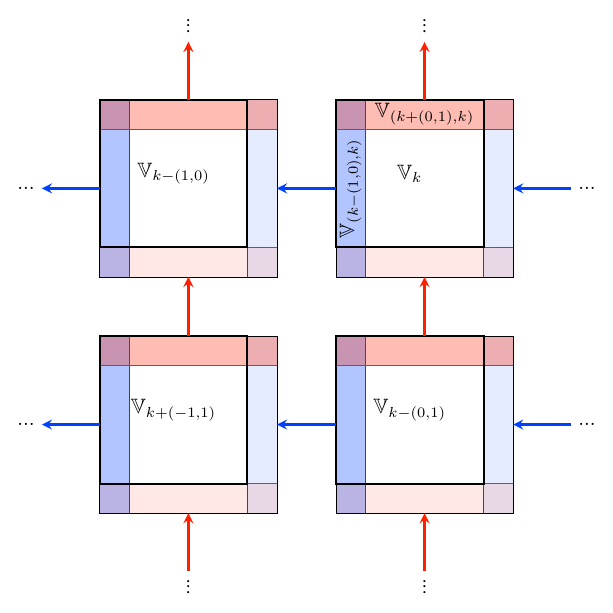}
    &	\includegraphics[keepaspectratio, width=0.5\textwidth]{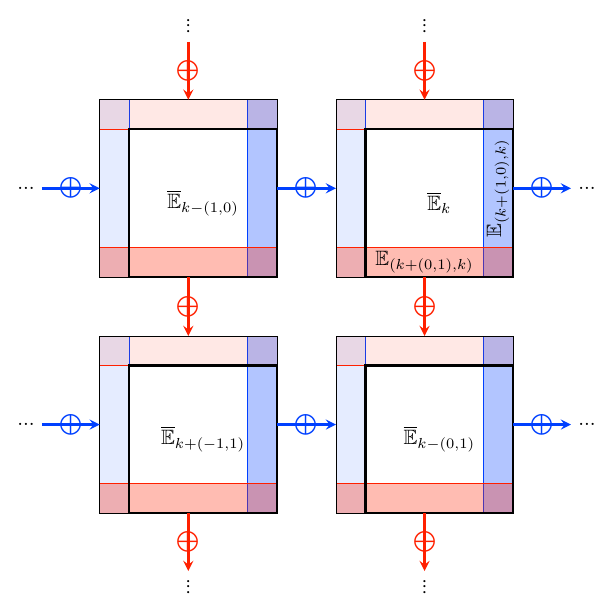}	\\[-0.2cm]
    (a) Distributed implementation of $\Db_1$
    &    (b) Distributed implementation of $\Db_1^*$
    \end{tabular}
    
    \vspace*{-0.2cm}
    \caption{\small
    Communication patterns involved in the distributed implementation of (a) $\Db_1$ and (b) $\Db_1^*$. 
    Each worker $k$ is required to communicate with two contiguous workers along each direction successively. 
    Colored arrows indicate whether vertices (a) or hyperedge weights (b) are received from or sent to a nearby worker. Vertices (a) and hyperedge weights (b) sent and received along the horizontal and vertical communication steps are highlighted in blue and red, respectively. (b) The circled ``plus'' symbols represent contributions aggregated with the corresponding hyperedge weights of the reception worker. These operations correspond to the second term in~\eqref{eq:aggregate_for_adjoint}.
    }
    \label{fig:app:comm2D}
\end{figure}

\paragraph{Distributed implementation of $\Db_1$} 

\Cref{fig:app:comm2D} illustrates the preliminary communications required by the distributed implementation of $\Db_1$. Each worker $k\in \{1, \ldots, K\}$ needs to collect $(\xb_{1,(k,k')})_{k' \in \Rc_k}$, as illustrated in~\Cref{fig:app:comm2D}(a). A first communication step occurs along the horizontal axis. Each worker $k$ sends vertices from its left-hand-side border (of width $L_2$, the horizontal width of the blur kernel) to its neighbour on the left (\Cref{fig:app:comm2D}(a), dark blue areas), and receives vertices from the neighbour on its right (\Cref{fig:app:comm2D}(a), light blue areas). Once the first step is complete, a second communication step occurs along the vertical axis with the top and bottom neighbours of the worker $k$ (\Cref{fig:app:comm2D}(a), dark and light red areas, of width $L_1$). The local operator $(\Db_{1,m,\overline{\Wc}_{1,m}})_{ m \in  \eE_{1,\Rc_k} }$ is equivalent to considering a convolution matrix and a selection operator. The latter ensures that  the correct boundaries are considered (only the convolution outputs which have not interacted with the boundaries of $(\xb_{\overline{\Wc}_{1, m}})_{ m \in  \eE_{\Rc_k} }$ are valid).

\paragraph{Distributed implementation of $\Db_1^*$} 

The distributed implementation of $\Db_1^*$ is similar to that of $\Db_1$. A first communication step occurs along the horizontal axis. Each worker $k$ sends hyperedge weights from its right-hand-side border (of width $L_2$) to its neighbour on the right (\Cref{fig:app:comm2D}(b), dark blue areas), and aggregates weights from the neighbour on its left (\Cref{fig:app:comm2D}(b), light blue areas). Similar communications occurs along the vertical axis with the bottom and top neighbours of the worker $k$ (\Cref{fig:app:comm2D}(b), dark and light red areas, of width $L_1$). A convolution is then applied on each worker to all the weights locally available, followed by a selection operator to ensure correct boundaries are used.

\paragraph{Distributed implementation of $\Db_2$}

The distributed implementation of $\Db_2$ requires the same communication pattern as $\Db_1$, successively exchanging messages with width 1 along the horizontal and vertical directions. The operator $(\Db_{2,m,\overline{\Wc}_{2,m}})_{ m \in  \eE_{2,\Rc_k} }$ corresponds to a local discrete gradient operator, using the boundaries retrieved during the communication step. Note that, for $k \in \{1, \dotsc, K\}$ and $k' \in \Rc_k$, the set of vertices $\eV_{2, (k,k')}$ to be communicated by the worker $k'$ to $k$ is such that $\eV_{2, (k,k')} \subset \eV_{1, (k,k')}$. The total number of elements to be communicated in the algorithm is thus reduced, as vertices required by the distributed implementation of $\Db_2$ already need to be communicated for $\Db_1$.

\paragraph{Distributed implementation of $\Db_2^*$}

The distributed implementation of $\Db_2^*$ requires the same communication pattern as $\Db_1^*$, successively exchanging messages with width 1 along the horizontal and vertical directions. A local adjoint discrete gradient operator is then applied on each worker $k$ to all the hyperedge weights available locally.

\subsection{Experiments}
\label{Ssec:application:experiments}

The proposed approach is evaluated in terms of estimation quality and scalability on the deconvolution problem of~\Cref{Ssec:application:problem_statement}. Results are compared with those of the reference serial SGS algorithm from~\cite{Vono2019icassp}.

\subsubsection{Simulation setting} \label{Ssec:experiments:settings}

All the experiments have been conducted on a single computer equipped with two 2.1~GHz, 18-core, Intel Xeon E5-2695 v4 series processors (36 CPU cores in total). In this setting, a worker corresponds to a process running on one CPU core. The proposed distributed sampler has been implemented in Python using the \texttt{mpi4py} library~\cite{Dalcin2021}. Codes to reproduce the experiments are available at \url{https://gitlab.cristal.univ-lille.fr/pthouven/dsgs}.

The proposed approach is compared with the SGS algorithm proposed in~\cite{Vono2019icassp}. The latter relies on a different splitting strategy compared to \Cref{Ssec:application:problem_statement}, using $\ncomposite = 3$ operators (see~\cite{Vono2019icassp} for further details).
In practice, MYULA transition kernels are leveraged to sample from conditional distributions involving non-smooth potential functions. This choice of transition kernel and splitting strategy requires the proximity operator of the TV norm to be evaluated at each iteration of the sampler, using a primal-dual algorithm~\cite{Chambolle2011}.

Performance is assessed in terms of average runtime per iteration (with associated standard deviation) and quality of both the \emph{minimum mean square error} (MMSE) and \emph{maximum a posteriori} (MAP) estimators. The estimators are denoted $\xmmse$ and $\xmap$, respectively. Reconstruction quality is quantified with the structural similarity index (SSIM)~\cite{Wang2004} and the signal-to-noise ratio (SNR)
expressed in dB. Associated 95\% credibility intervals (CIs) are also reported.

The sampler from~\cite{Vono2019icassp} has been applied with $\kappa = 1$ and $((\alpha_\idcomposite^2, \beta_\idcomposite^2)_{1 \leq \idcomposite \leq 3}) = \mathbf{1}_6$, where $\mathbf{1}_Q \in \mathbb{R}^Q$ is a vector with entries all equal to 1. The proposed approach uses $\kappa = 1$ and $((\alpha_\idcomposite^2, \beta_\idcomposite^2)_{1 \leq \idcomposite \leq 2}) = \mathbf{1}_4$. For both algorithms, $\nmc = 5 \times 10^3$ samples have been generated to form $\xmmse$, $\xmap$ and the 95\% CIs, discarding $\nbi = 2 \times 10^3$ burn-in samples.

\subsubsection{Experiment results} \label{Ssec:experiments:results}

\paragraph*{Estimation quality}

Using $K=1$ worker, the proposed approach is compared with~\cite{Vono2019icassp}. Ground truth images with different maximum intensity levels $\overline{x}_{\max} = \max_{1 \leq n \leq N} \overline{x}_n$ have been considered. The values $\overline{x}_{\max} \in \{20, 30\}$ have been adopted for the following datasets:
\begin{enumerate}
    \item \texttt{house} image ($N = 256^2$) with a normalized Gaussian kernel of size $L \in \{3^2, 7^2\}$;
    \item \texttt{peppers} image ($N = 512^2$) with a normalized Gaussian kernel of size $L \in \{7^2, 15^2\}$.
\end{enumerate}

The results reported in \Cref{tab:results_serial} show that the estimators formed with the proposed approach have higher quality metrics compared to the method proposed in~\cite{Vono2019icassp}. In addition, the computing time required by the proposed sampler is between 1.5 and 2 times smaller than~\cite{Vono2019icassp}. This discrepancy comes from the difference in the splitting strategies considered by the two methods. In particular, the splitting approach proposed in~\cite{Vono2019icassp} requires the evaluation of the proximal operator of the TV norm, obtained as the output of an iterative optimization algorithm. The difference in splitting can also affect the quality of the resulting AXDA approximation, as can be seen in the difference in quality of the estimators. The MMSE estimator reported in \Cref{fig:results:peppers30} for \cite{Vono2019icassp} appears much smoother compared to the proposed algorithm for the same regularization parameter $\kappa$.
Note that the uncertainty level of the proposed approach is slightly lower and appears more diffuse than~\cite{Vono2019icassp}.

\begin{table}[htbp]
    \centering
    \resizebox{0.98\textwidth}{!}{%
    \begin{tabular}{lllrrrrrrrr} \toprule
        & Dataset & Algo. & $\SNR(\xmmse)$ & $\SNR(\xmap)$ & $\SSIM(\xmmse)$ & $\SSIM(\xmap)$ & Time per iter.  & Runtime \\
        &&&&&&& ($\times 10^{-2}$ s) & ($\times 10^{2}$ s) \\ \midrule
        \multirow{9}{*}{\rotatebox{90}{$\overline{x}_{\max} = 20$}} & House & \cite{Vono2019icassp} & 18.37 & 15.08 & 0.60 & 0.20 & 5.59 (0.12) & 1.68 \\
        & ($L=3^2$) & Proposed & 20.21 & 16.18 & 0.60 & 0.26 & \textbf{2.49 (0.09)} & \textbf{0.75} \\ \cmidrule{2-9}
        & House & \cite{Vono2019icassp} & 18.00 & 14.92 & 0.60 & 0.20 & 8.52 (0.22) & 2.56 \\
        & ($L=7^2$) & Proposed & 19.86 & 15.98 & 0.59 & 0.24 & \textbf{4.88 (0.13)} & \textbf{1.47} \\ \cmidrule{2-9}
        & Peppers & \cite{Vono2019icassp} & 18.98 & 15.08 & 0.66 & 0.25 & 16.71 (0.38) & 5.01 \\
        & ($L=7^2$) & Proposed & 20.52 & 16.07 & 0.67 & 0.30 & \textbf{7.73 (0.33)} & \textbf{2.32} \\ \cmidrule{2-9}
        & Peppers & \cite{Vono2019icassp} & 18.90 & 15.07 & 0.66 & 0.25 & 36.68 (0.33) & 11.01 \\
        & ($L=15^2$) & Proposed & 20.52 & 16.04 & 0.66 & 0.30 & \textbf{23.32 (0.40)} & \textbf{7.00} \\
        \midrule\midrule
        \multirow{9}{*}{\rotatebox{90}{$\overline{x}_{\max} = 30$}} & House & \cite{Vono2019icassp} & 18.23 & 16.46 & 0.64 & 0.29 & 4.71 (0.10) & 1.41 \\
        & ($L=3^2$) & Proposed & 20.18 & 17.84 & 0.66 & 0.34 & \textbf{2.03 (0.05)} & \textbf{0.61} \\ \cmidrule{2-9}
        & House & \cite{Vono2019icassp} & 17.89 & 16.25 & 0.64 & 0.28 & 7.01 (0.10) & 2.10 \\
        & ($L=7^2$) & Proposed & 19.80 & 17.55 & 0.65 & 0.33 & \textbf{3.93 (0.10)} & \textbf{1.18} \\ \cmidrule{2-9}
        & Peppers & \cite{Vono2019icassp} & 19.03 & 16.77 & 0.69 & 0.35 & 17.04 (0.45) & 5.11 \\
        & ($L=7^2$) & Proposed & 20.71 & 17.96 & 0.71 & 0.40 & \textbf{10.26 (1.22)} & \textbf{3.08} \\ \cmidrule{2-9}
        & Peppers & \cite{Vono2019icassp} & 19.00 & 16.72 & 0.69 & 0.35 & 36.81 (0.44) & 11.04 \\
        & ($L=15^2$) & Proposed & 20.74 & 17.98 & 0.71 & 0.41 & \textbf{23.22 (0.34)} & \textbf{6.97} \\
        \bottomrule
    \end{tabular}
    }
    \caption{\small
    Comparison between~\cite{Vono2019icassp} and the proposed approach with $\nworkers = 1$. Datasets have been generated from ground truth images with different maximum intensity $\overline{x}_{\max} \in \{20, 30\}$ and convolution kernel sizes $L \in \{3, 7, 15\}$. Results are reported in terms of estimation quality, average runtime per iteration (with standard deviation) and total runtime.}
    \label{tab:results_serial}
\end{table}

\begin{figure}[htbp]
    \captionsetup[subfigure]{justification=centering}
    \centering
    \small
    \begin{subfigure}[t]{0.31\textwidth}
         \centering
         \includegraphics[keepaspectratio, width=0.95\textwidth]{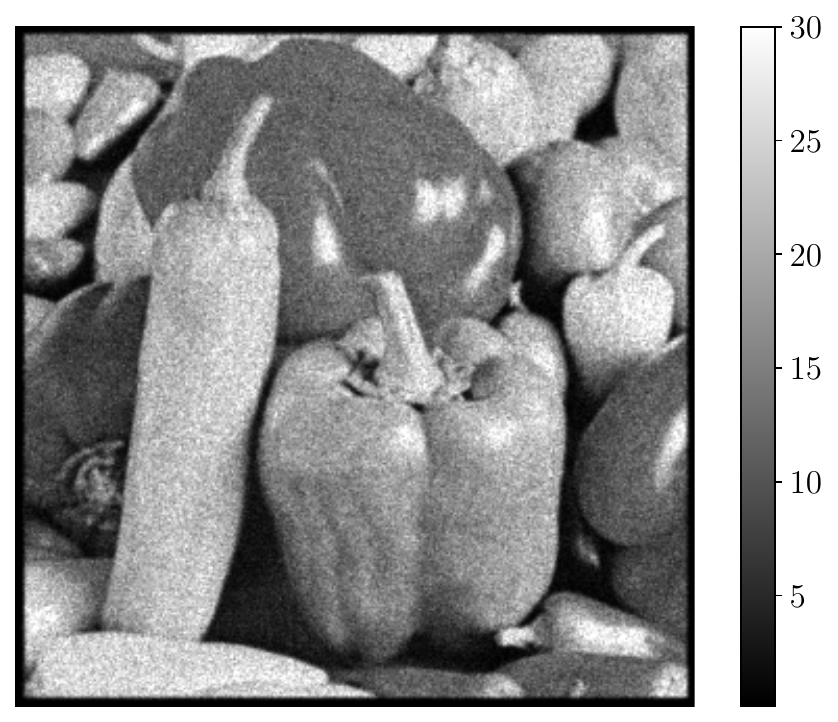} 
         \includegraphics[keepaspectratio, width=0.95\textwidth]{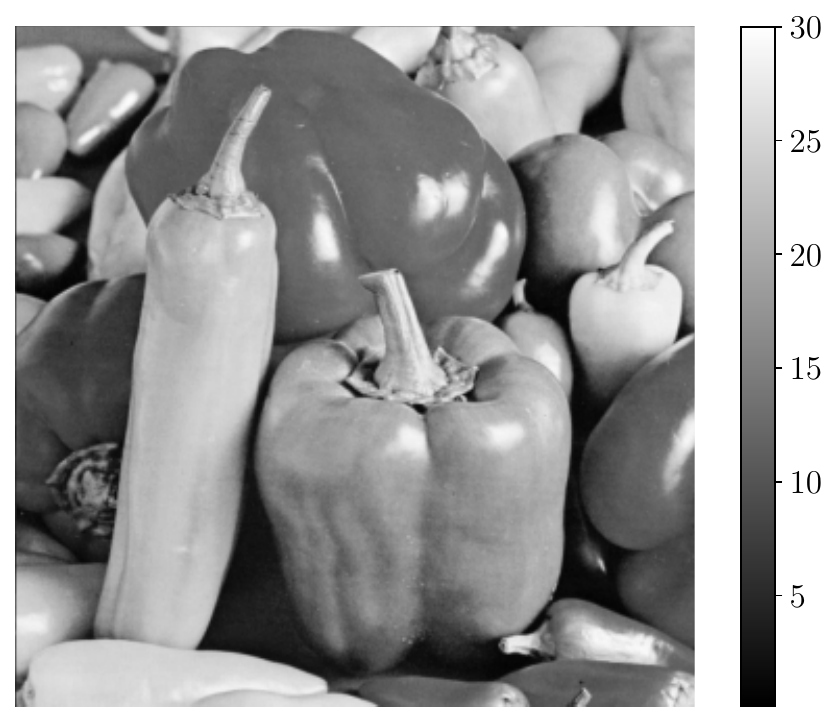}
         \caption{Observations and \\ground truth}
     \end{subfigure}
     \begin{subfigure}[t]{0.31\textwidth}
         \centering
         \includegraphics[keepaspectratio, width=0.95\textwidth]{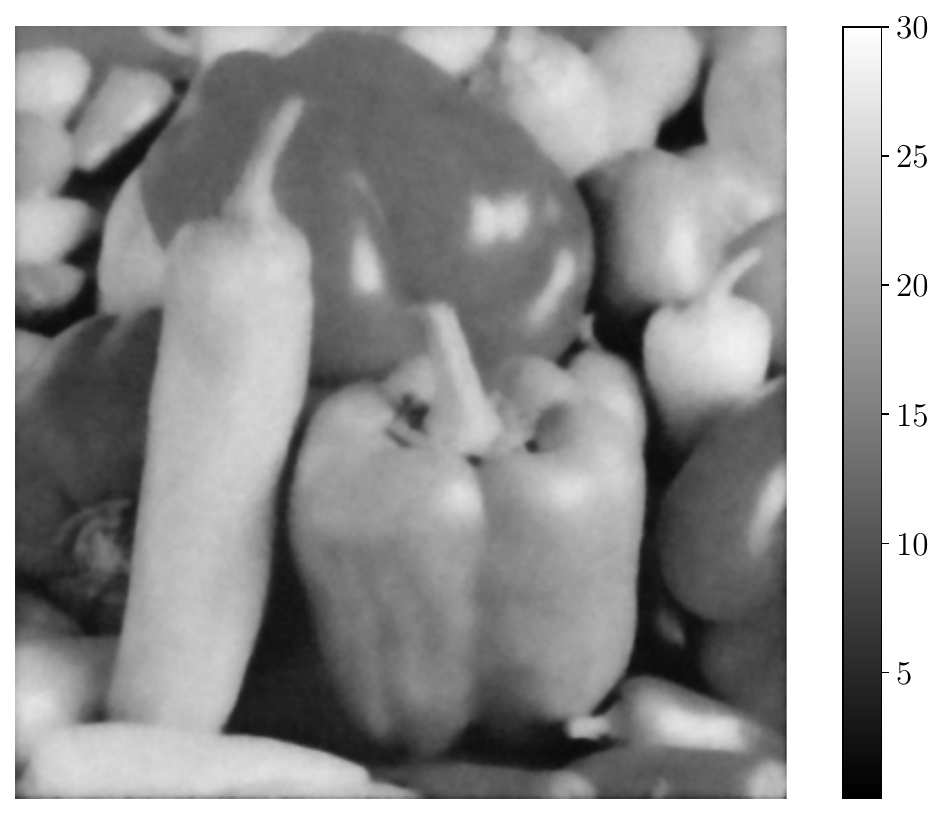}
         \includegraphics[keepaspectratio, width=0.95\textwidth]{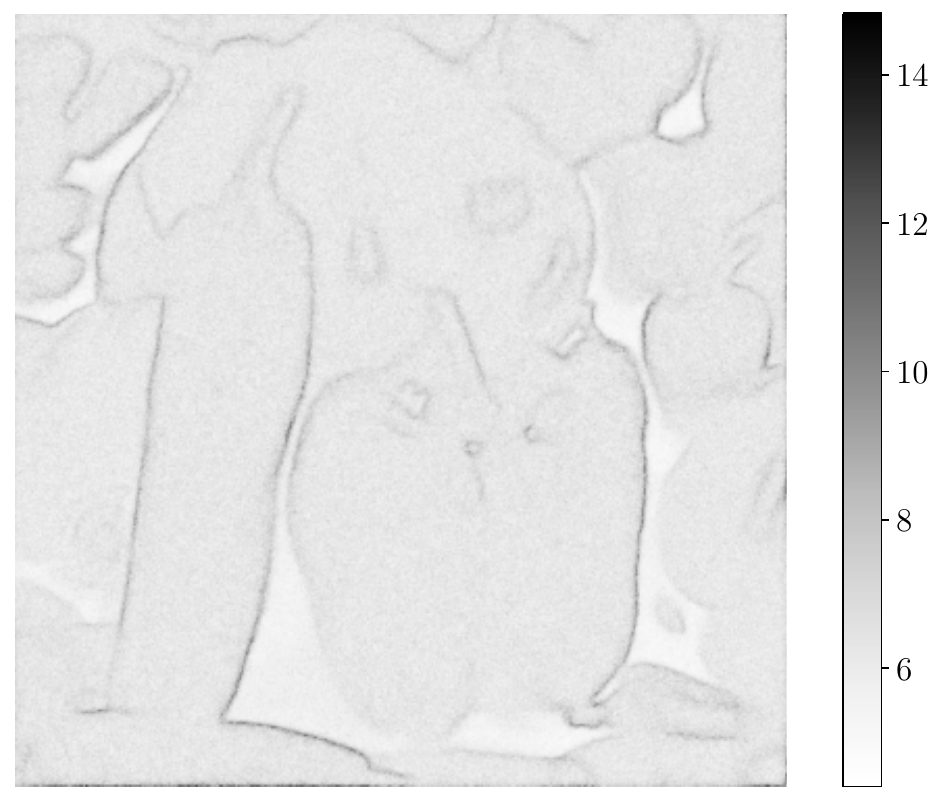}
         \caption{MMSE estimator and \\ 95\% CI~\cite{Vono2019icassp}}
     \end{subfigure}
     \begin{subfigure}[t]{0.31\textwidth}
         \centering
         \includegraphics[keepaspectratio, width=0.95\textwidth]{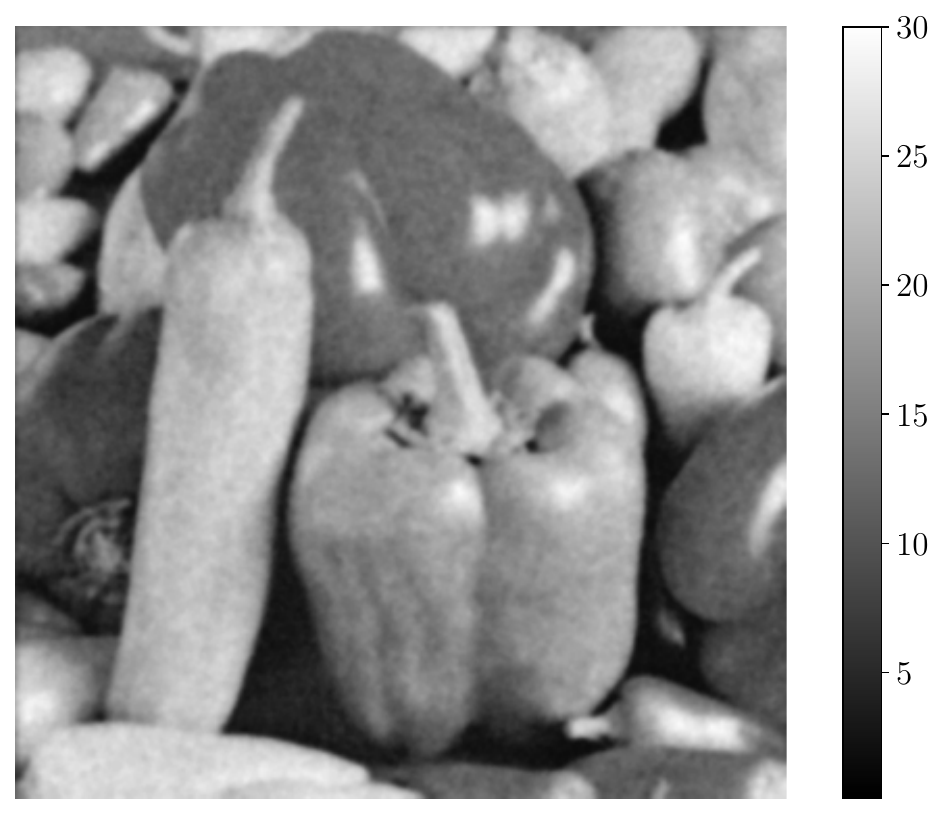}
         \includegraphics[keepaspectratio, width=0.94\textwidth]{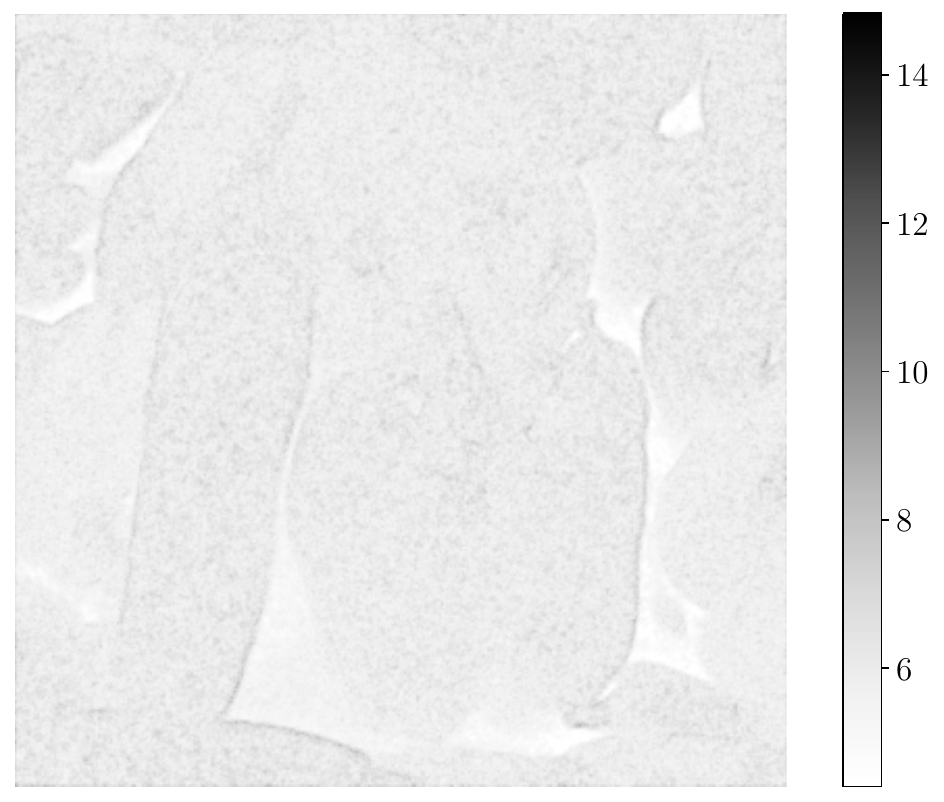}
         \caption{MMSE estimator and \\ 95\% CI~(prop., $K=1$)}
     \end{subfigure}
    
    \vspace{-0.5cm}
    
    \caption{\small
    Estimators and 95\% credibility intervals for \texttt{peppers}, with $\overline{x}_{\max} = 30$ and $L=15^2$.}
    \label{fig:results:peppers30}
\end{figure}

\paragraph*{Strong scaling experiment}

For this experiment, the behaviour of the proposed method is investigated with a varying number of workers $K \in \{1, 2, 4, 8, 16, 32\}$, using the maximum intensity level $\overline{x}_{\max} = 30$ for the \texttt{peppers} dataset ($N = 512^2$).

\Cref{tab:results_strong_scaling_small_kernel,tab:results_strong_scaling_large_kernel} show the speedup for the proposed approach when $L=7^2$ and $L=15^2$, respectively. 
In \Cref{tab:results_strong_scaling_large_kernel}, a close to ideal speedup (\emph{i.e.}, close to the number of workers $K$) is observed as $L=15^2$. Note that speedup factors larger than the number of cores may occur, depending on the cache state of the machine at the time the experiments have been run. In comparison, the ideal runtime per iteration (i.e., neglecting communication time) for~\cite{Vono2019icassp} for a client-server architecture with $K = 4$ cores is $1.48 \times 10^{-1}$ s when $L=7^2$ ($1.04 \times 10^{-1}$ s for TV-related terms), and $3.41 \times 10^{-1}$ s when $L=15^2$ ($1.02 \times 10^{-1}$ s for TV-related terms). In both cases, the time required to update the other splitting variables is almost 10 times lower. For this experiment, a client server approach thus leads to a limited runtime performance, given the heterogeneity in the complexity of the tasks assigned to workers. All these results illustrate the ability of the proposed sampler to provide estimators at a fraction of the runtime of the serial implementation by using an increasing number of cores $K$.

\paragraph*{Weak scaling experiment}

The behaviour of the proposed method is investigated when both the size of the problem and the number of workers considered simultaneously increase.
A fixed problem size per worker is considered, using $K \in \{1, 4, 16\}$. Datasets derived from upsampled versions of the \texttt{house} image are considered, with a maximum intensity level $\overline{x}_{\max} = 30$, using $(N, L) \in \{(256^2, 3^2),$ $(512^2, 7^2), (1022^2, 11^2)\}$. These numbers ensure that both the problem size $M$ and the associated number of workers $K$ evolve in the same proportions over the configurations tested.

\Cref{tab:results_weak_scaling} shows scaled speedup factors (that is, normalized by the factor of increase for $K$ and $M$) close to the number of cores used. Differences from a linear scaling may result from fixed communication costs, representing a larger cost per iteration as $L$ increases. Overall, the results illustrate the runtime stability of the approach for a fixed problem size per worker. This experiment efficiently processes a 1 million pixel image in about 2 minutes to obtain an estimator with the associated credibility intervals.

\begin{table*}[t]
    \centering
    \resizebox{0.98\textwidth}{!}{%
    \begin{tabular}{crrrrrrr} \toprule
        $\nworkers$ & $\SNR(\xmmse)$ & $\SNR(\xmap)$ & $\SSIM(\xmmse)$ & $\SSIM(\xmap)$ & Time per iter.  & Speedup & Runtime  \\
        &&&&& ($\times 10^{-2}$ s) && ($\times 10^{2}$ s) \\ \midrule
        1 & 20.71 & 17.96 & 0.71 & 0.40 & 10.26 (1.22) & \textbf{1.00} & 3.08 \\
        2 & 20.70 & 17.94 & 0.71 & 0.40 & 5.33 (0.10) & \textbf{1.93} & 1.60 \\
        4 & 20.72 & 17.93 & 0.71 & 0.40 & 4.35 (0.12) & \textbf{2.36} & 1.30 \\
        8 & 20.72 & 17.95 & 0.71 & 0.40 & 2.51 (0.20) & \textbf{4.08} & 0.75 \\
        16 & 20.73 & 17.97 & 0.71 & 0.41 & 1.23 (0.03) & \textbf{8.32} & 0.37 \\
        32 & 20.71 & 17.91 & 0.71 & 0.40 & \textbf{0.60 (0.06)} & \textbf{17.04} & \textbf{0.18} \\
      \bottomrule
    \end{tabular}
    }
    \caption{\small
    Results of the strong scaling experiment using a dataset with $\overline{x}_{\max} = 30$ and kernel size $L=7^2$. Performance is reported in terms of estimation quality, time per iteration and speedup.} 
    \label{tab:results_strong_scaling_small_kernel}
\end{table*}

\begin{table*}[thbp]
    \centering
    \resizebox{0.98\textwidth}{!}{%
    \begin{tabular}{crrrrrrr} \toprule
        $\nworkers$ & $\SNR(\xmmse)$ & $\SNR(\xmap)$ & $\SSIM(\xmmse)$ & $\SSIM(\xmap)$ & Time per iter.  & Speedup & Runtime  \\
        &&&&& ($\times 10^{-2}$ s) && ($\times 10^{2}$ s) \\ \midrule
        1 & 20.74 & 17.98 & 0.71 & 0.41 & 23.22 (0.34) & \textbf{1.00} & 6.97 \\
        2 & 20.72 & 17.94 & 0.71 & 0.40 & 12.42 (0.11) & \textbf{1.87} & 3.73 \\
        4 & 20.74 & 17.98 & 0.71 & 0.41 & 3.22 (0.10) & \textbf{7.22} & 0.97 \\
        8 & 20.74 & 17.97 & 0.71 & 0.41 & 1.91 (0.06) & \textbf{12.17} & 0.57 \\
        16 & 20.72 & 17.95 & 0.71 & 0.40 & 1.31 (0.11) & \textbf{17.73} & 0.39 \\
        32 & 20.73 & 17.97 & 0.71 & 0.40 & \textbf{0.70 (0.07)} & \textbf{33.35} & \textbf{0.21} \\
      \bottomrule
    \end{tabular}
    }
    \caption{\small
    Results of the strong scaling experiment using a dataset with $\overline{x}_{\max} = 30$ and kernel size $L = 15^2$. Performance is reported in terms of estimation quality, time per iteration and speedup. 
    } 
    \label{tab:results_strong_scaling_large_kernel}
\end{table*}

\begin{table*}[thbp]
    \centering
    \resizebox{0.98\textwidth}{!}{%
    \begin{tabular}{@{}c@{}rrrrrrrr@{}} \toprule
        $(N, L,\nworkers)$  & $\SNR(\xmmse)$ & $\SNR(\xmap)$ & $\SSIM(\xmmse)$ & $\SSIM(\xmap)$ & Time per iter.  & Scaled & Runtime  \\
        &&&&& ($\times 10^{-2}$ s) & speedup & ($\times 10^{2}$ s) \\ \midrule
        $(256^2, 3^2, \mathbf{1})$  & 20.18 & 17.84 & 0.66 & 0.34 & 2.03 (0.05) & \textbf{1.00} & \textbf{0.61} \\
        $(512^2, 7^2, \mathbf{4})$  & 23.86 & 19.56 & 0.74 & 0.34 & 2.13 (0.08) & \textbf{3.81} & \textbf{0.64} \\
        $(1022^2, 11^2, \mathbf{16})$ & 27.55 & 20.65 & 0.81 & 0.33 & 3.56 (0.18) & \textbf{9.13} & \textbf{1.07} \\
      \bottomrule
    \end{tabular}
    }
    \caption{\small
    Results of the weak scaling experiment using a dataset with $\overline{x}_{\max} = 30$. The reconstruction quality is reported with the time per iteration, the scaled speedup and the runtime.}
    \label{tab:results_weak_scaling}
\end{table*}

\section{Conclusion}
\label{Sec:conclusion}

In this paper, a distributed block-coordinate SGS has been introduced to efficiently solve large scale imaging inverse problems. 
The approach leverages the approximate data augmentation scheme AXDA~\cite{Vono_etal_2020,Rendell2021} to efficiently handle composite functions involving linear operators. A block-coordinate approach is adopted to split and distribute all the variables over multiple workers. 
The proposed method exploits the hypergraph structure of the linear operators to design a versatile distributed block-coordinate split Gibbs sampler.

Experiments on a supervised image deblurring problem show that the proposed approach forms reliable estimates with quantified uncertainty in a significantly reduced amount of time, compared to a state-of-the-art non-distributed version of the sampler from~\cite{Vono2019icassp}. In particular, the proposed sampler is shown to provide estimators at a fraction of the runtime of the serial implementation by using an increasing number of workers $K$. Processing a 1 million pixel image using our current Python implementation takes less than 2 minutes to obtain a good restoration with associated credibility intervals.

Note that the proposed distributed block-coordinate SGS is directly applicable to a much wider class of applications than the restoration problems addressed in this work, such as image inpainting, super-resolution or reconstruction. Future works include the development of an asynchronous version of the proposed approach to further speed up the inference process, while maintaining the convergence of the Markov chain.

\if0\production
{   
    \appendix
\section{SGS for multiple composite terms}
\label{appendix:SPA-PSGLA-multi}

\subsection{AXDA for multiple composite terms}
\label{appendix:Ssec:AXDA-multi}

Following~\cite{Vono_etal_2020}, AXDA can be generalized to address distributions of the form~\eqref{eq:dist-gen-split}, which involve multiple composite terms. We will use similar ideas to generalize the proposed approach, leveraging the approximation~\eqref{eq:dist-gen-split-SPA} associated with distribution~\eqref{eq:dist-gen-split}.
In \cite{Vono_etal_2020}, the authors show that
Proposition~\ref{Prop:cv-SPA-simple} holds in this context, under the same conditions on $\phi_{\idcomposite, \alpha_\idcomposite}$ and $\psi_{\idcomposite, \beta_\idcomposite}$ as~\eqref{prop:cv-spa-simple-cond1}-\eqref{prop:cv-spa-simple-cond2}.

Using an approach similar to the one described in \Cref{Ssec:SPA-PSGLA}, we can design a PSGLA within Gibbs sampler to approximately draw samples from~\eqref{eq:dist-gen-split-SPA}, generalizing Algorithm~\ref{algo:SPA-simple-sampling}. The resulting algorithm is reported in Algorithm~\ref{algo:SPA-multi-sampling},
where $(\wb^{(t)})_{0 \le t \le T}$ is a sequence of independent and identically distributed (i.i.d) standard Gaussian random variables in $\HHH$, 
$\gamma \in ]0, \lambda^{-1}[$, 
$\lambda = \lambda_h + \| \sum_{\idcomposite=1}^\ncomposite \Db_\idcomposite/\alpha_\idcomposite \|^2$, and, for every $\idcomposite \in \{1, \ldots, \ncomposite\}$,
\begin{align}
    (\forall \idcomposite \in \{1, \ldots, \ncomposite\})\quad
       \pi_{\idcomposite, \alpha_\idcomposite}(\zb_\idcomposite \mid \vb_\idcomposite,\ub_\idcomposite ) 
    \propto \exp \left( -g_\idcomposite(\zb_\idcomposite) - \phi_{\idcomposite, \alpha_\idcomposite}(\vb_\idcomposite, \zb_\idcomposite-\ub_\idcomposite)  \right) , \label{eq:dist-SPA-split-cond:zi}
\end{align}
with $\vb_\idcomposite = \Db_\idcomposite \xb$.

\begin{algorithm}[thbp]
    \small
    \KwIn{$\xb^{(0)} \in \HHH$, $(\zb_\idcomposite^{(0)}, \ub_\idcomposite^{(0)}) \in \GGG_\idcomposite^2$ and $(\alpha_\idcomposite, \beta_\idcomposite) \in ]0,+\infty[^2$ for $\idcomposite \in \{1, \dotsc, \ncomposite\}$, $\gamma \in \big]0, \big(\lambda_h + \| \sum_{\idcomposite=1}^\ncomposite \Db_\idcomposite/\alpha_\idcomposite \|^2 \big)^{-1} \big[$}

    \For{$t = 0$ \KwTo $T$}{

        $\displaystyle \xb^{(t+1)} = \prox_{\gamma f} \left( \xb^{(t)} - \gamma \nabla h(\xb^{(t)}) - \gamma \sum_{\idcomposite=1}^\ncomposite \Db_\idcomposite^* \nabla \phi_{\idcomposite,\alpha_\idcomposite}(\cdot, \zb_\idcomposite^{(t)}-\ub_\idcomposite^{(t)})(\vb_\idcomposite^{(t)}) + \sqrt{2\gamma} \wb^{(t)} \right)$; \\

        \For{$i = 1$ \KwTo $\ncomposite$}{
                $\displaystyle \vb_\idcomposite^{(t+1)} = \Db_\idcomposite \xb^{(t+1)}$, \\
                $\displaystyle \zb_\idcomposite^{(t+1)} \sim 
                \distc{\idcomposite, \alpha_\idcomposite}{\zb_\idcomposite}{ \vb_\idcomposite^{(t+1)}, \ub_\idcomposite^{(t)}}$; \\
                $\displaystyle \ub_\idcomposite^{(t+1)} \sim \Nc \Big( \frac{\beta_\idcomposite^2}{\alpha_\idcomposite^2+\beta_\idcomposite^2} (\zb_\idcomposite^{(t+1)} - \vb_\idcomposite^{(t+1)}), \frac{\alpha_\idcomposite^2+\beta_\idcomposite^2}{\alpha_\idcomposite^2 \beta_\idcomposite^2} \Idb \Big)$;
        }
    }
    \KwOut{$(\xb^{(t)})_{1 \le t \le T}$, $(\zb_\idcomposite^{(t)}, \ub_\idcomposite^{(t)})_{1 \le t \le T, 1 \leq i \leq \ncomposite}$}
	\caption{\small Proposed multi-term SGS (with $\ncomposite \geq 1$ composite terms).}
	\label{algo:SPA-multi-sampling}
\end{algorithm}

\begin{table}[t!]
    \centering\footnotesize
    \begin{tabular}{@{}p{4.2cm}@{}|p{7.5cm}|@{}c@{}|c@{}c@{}} \toprule
        \multirow{2}{*}{Notation} & \multirow{2}{*}{Definition} & \multicolumn{2}{c}{Given by}    \\
        && operator $\Db_\idcomposite$ & $\phantom{a.}$ user $\phantom{.a}$    \\
        \hline
        $\eb_\idcomposite = (e_{\idcomposite,m})_{1 \le m \le M_\idcomposite}$ & Hyperedges of $\Hc_\idcomposite$ & \checkmark & \\
        $e_{\idcomposite,m} \subset \{1, \ldots, N\}$ & Vertex indices in hyperedge $m$ of $\Hc_\idcomposite$ & \checkmark & \\
        $k_{\idcomposite,m} \in\{1, \ldots, K\}$ & Worker associated with $m$-th hyperedge $e_{\idcomposite,m}$ (chosen by the user). $k_{\idcomposite,m}$ must satisfy $e_{\idcomposite,m} \cap \eV_{k_{\idcomposite,m}} \neq \emp$  & & \checkmark \\
        $\Wc_{\idcomposite,m} \subset \{1, \ldots, K\} \setminus \{k_{\idcomposite,m}\}$ & Set of all workers but $k_{\idcomposite,m}$, containing vertices from $e_{\idcomposite,m}$ & \checkmark & \\
        $\overline{\Wc}_{\idcomposite,m} \subset \{1, \ldots, K\} $ & $\overline{\Wc}_{\idcomposite,m} = k_{\idcomposite,m} \cup \Wc_{\idcomposite,m}$ & \checkmark & \\
        $\eV_{\idcomposite,(k,k')} \subset \eV_{k'}$ & Indices of vertices sent from worker $k'$ to worker~$k$ & \checkmark & \checkmark    \\
        $\eV_{\Wc_{\idcomposite,m}} \subset \{1, \ldots, N\}\setminus \eV_{k_{\idcomposite,m}}$ & $\eV_{\Wc_{\idcomposite,m}} = \bigcup_{k' \in \Wc_{\idcomposite,m}} \eV_{(k_{\idcomposite,m}, k')}$ the set of vertex indices that will be communicated to worker $k_{\idcomposite,m}$ from all workers $k'\in \eV_{\idcomposite,m}$ & \checkmark & \checkmark    \\
        $\eV_{\overline{\Wc}_{\idcomposite,m}} \subset \{1, \ldots, N\}$ & $\eV_{\overline{\Wc}_{\idcomposite,m}} = \eV_{k_{\idcomposite,m}} \cup \eV_{\Wc_{\idcomposite,m}}$ the set of vertex indices necessary to perform computations associated with $k_{\idcomposite,m}$ & \checkmark & \checkmark    \\
        $\eE_{\idcomposite,k} \subset \{1, \ldots, M_\idcomposite\}$ & Indices of hyperedges only containing vertices stored on worker~$k$ & \checkmark & \checkmark \\
        $\eE_{\idcomposite,(k,k')} \subset \{1, \ldots, M_\idcomposite\}$ & Indices of hyperedges containing vertices sent from worker $k'$ to worker~$k$ & \checkmark & \checkmark \\
        $\eE_{\idcomposite,\Rc_k} \subset \{1, \ldots, M_\idcomposite\}$ & $\eE_{\idcomposite,\Rc_k} = \cup_{k'\in \Rc_k} \eE_{\idcomposite,(k,k')}$ set of all hyperedges containing vertices communicated to worker $k$ & \checkmark & \checkmark    \\
        $\overline{\eE}_{\idcomposite,k} \subset \{1, \ldots, M_\idcomposite\}$ & $ \overline{\eE}_{\idcomposite,k} = \eE_{\idcomposite,k} \cup \eE_{\idcomposite,\Rc_k} $, such that $(\overline{\eE}_{\idcomposite,k})_{1 \le k \le K}$ is a partition of $\{1, \ldots, M_\idcomposite\}$ & \checkmark & \checkmark \\
        \bottomrule
    \end{tabular}
    \caption{\label{table:notation-hypergraph-multi}\small
    Summary of the set notation used to define the hypergraph structure associated with the operators $\Db_\idcomposite$, for $\idcomposite \in \{1, \ldots, \ncomposite\}$.
    This notation generalizes the one from Table~\ref{table:notation-hypergraph}.
    }
\end{table}

\begin{table}[thbp]
    \centering\footnotesize
    \begin{tabular}{@{}p{4.9cm}|p{10.4cm}@{}} \toprule
        Notation & Definition \\
        \hline
        $\Db_\idcomposite \colon \HHH \to \GGG_\idcomposite$ & $\Db_\idcomposite = (D_{\idcomposite,m,n})_{1 \le m \le M_\idcomposite, 1\le n \le N}$ linear operator defining the $i$th hypergraph \\
        $\ub_\idcomposite, \vb_\idcomposite \in \GGG_\idcomposite$ & \multirow{2}{*}{Hyperedge weights} \\
        $\vb_\idcomposite = (v_{\idcomposite,m})_{1 \le m \le M_\idcomposite} \in \GGG_\idcomposite$ & \\
        $\vb_{\idcomposite,k} \in \GGG_{\idcomposite,k}$ & $\vb_{\idcomposite,k} = (v_{\idcomposite,m})_{m \in \overline{\eE}_{\idcomposite,k}}$ hyperedge weights stored on worker $k$ \\  
        $\Db_{\idcomposite,m,k} \colon \HHH_{\idcomposite,k} \to \GG_{\idcomposite,m} $ & $\Db_{\idcomposite,m,k} = (D_{\idcomposite,m,n})_{n \in \eV_{\idcomposite, k}}$, for $m\in \eE_{\idcomposite,k}$, subpart of $\Db_\idcomposite$ stored on worker $k$, associated with hyperedges containing vertices only on worker $k$ \\
        $\xb_{\idcomposite,(k,k')} \in \underset{n \in \eV_{\idcomposite, (k,k')}}{\Cart} \HH_n$  & $\xb_{\idcomposite,(k,k')} = (x_n)_{n \in \eV_{\idcomposite,(k,k')}}$ vertex values sent from worker $k'$ to $k$ \\
        $\Db_{\idcomposite,m,(k,k')} \colon \!\!\!\!\!\! \underset{n\in  \eV_{\idcomposite,(k,k')}}{\Cart} \!\!\!\!\! \HH_n \to \GG_{\idcomposite,m} $ & $\Db_{\idcomposite,m,(k,k')} = (D_{\idcomposite,m,n})_{n \in \eV_{\idcomposite,(k,k')}}$, for $m\in \eE_{\idcomposite,(k,k')}$, part of $\Db_\idcomposite$ stored on worker $k$, related to hyperedges containing vertices overlapping $k$ and $k'$ \\
        $\xb_{\overline{\Wc}_{\idcomposite,m}} \in \underset{n\in \eV_{\overline{\Wc}_{\idcomposite,m}}}{\Cart} \HH_n$  & $\xb_{\overline{\Wc}_{\idcomposite,m}} = (x_n)_{n \in \eV_{\overline{\Wc}_{\idcomposite,m}}}$, concatenation of vertex values stored on worker $k_{\idcomposite,m}$, and those sent from all workers $k' \in \Wc_{\idcomposite,m}$ to $k_{\idcomposite,m}$ \\
        $\Db_{\idcomposite,m,\overline{\Wc}_{\idcomposite,m}} \colon \underset{n\in \eV_{\overline{\Wc}_{\idcomposite,m}} }{\Cart} \HH_n \to \GG_{\idcomposite,m} $ & $\Db_{\idcomposite, m,\overline{\Wc}_{\idcomposite,m}} = (D_{\idcomposite,m,n})_{n \in \eV_{\overline{\Wc}_{\idcomposite,m}}}$, for $m\in \eE_{\idcomposite,\Rc_k}$, subpart of $(D_{\idcomposite,m,n})_{1 \le n \le N}$ corresponding to vertices stored on worker $k_{\idcomposite,m}$ and vertices overlapping worker $k_{\idcomposite,m}$ and other workers of $\Wc_{\idcomposite,m}$ \\
        \bottomrule
    \end{tabular}
    \caption{\label{table:notation2-hypergraph-multi}\small
    Notation used for the variables associated with the hypergraph induced by the operators $\Db_\idcomposite$, for $\idcomposite \in \{1, \ldots, \ncomposite\}$. This notation generalizes the one provided in Table~\ref{table:notation2-hypergraph}.
    }
\end{table}

\subsection{Distributed multi-term SGS}
\label{appendix:Ssec:SPA-PSGLA-multi-algo}

This section introduces a distributed version of Algorithm~\ref{algo:SPA-multi-sampling}, using the same approach as in Sections~\ref{Sec:Model} and \ref{Sec:dist-single-algo}. 
We consider the distribution~\eqref{eq:dist-gen-split}, with its AXDA approximation~\eqref{eq:dist-gen-split-SPA}.

For every $\idcomposite \in \{1, \ldots, \ncomposite\}$, let $\GGG_\idcomposite = \RR^{\overline{M}_\idcomposite}$ such that $\GGG_\idcomposite = \GG_{\idcomposite,1} \times \ldots \times \GG_{\idcomposite,M_\idcomposite}$, where for every $m\in \{1, \ldots, M_\idcomposite\}$, $\GG_{\idcomposite,m} = \RR^{M_{\idcomposite,m}}$, with $\overline{M}_\idcomposite = \sum_{m=1}^{M_\idcomposite} M_{\idcomposite,m}$.
An element of $\GGG_\idcomposite$ is denoted by $\ub_\idcomposite = (u_{\idcomposite,m})_{1 \le m \le M_\idcomposite}$, where, for every $m \in \{1, \ldots, M_\idcomposite\}$, $u_{\idcomposite,m} \in \GG_{\idcomposite,m}$.

For every $\idcomposite \in \{1, \ldots,\ncomposite\}$, the linear operator $\Db_\idcomposite = (D_{\idcomposite,m,n})_{1 \le m \le M_\idcomposite, 1 \le n \le N}$ defines a hypergraph structure $\Hc_\idcomposite$ as described in \Cref{Sec:Model}. We thus consider $\ncomposite$ hypergraphs, distributed over the same $K$ workers. As in \Cref{Sec:Model}, the choice of the distribution is made by the user, depending on the shape of the hypergraphs. The associated notation given in \Cref{table:notation-hypergraph-multi,table:notation2-hypergraph-multi} generalize those introduced in \Cref{Sec:Model}. For every $k \in \{1, \ldots, K\}$, let $\GGG_{\idcomposite, k} = \cart_{m \in \overline{\eE}_{\idcomposite, k}} \GG_{\idcomposite,m}$ such that $\GGG_\idcomposite = \cart_{1\leq k \leq K} \GGG_{\idcomposite, k}$.

We assume that, for every $\idcomposite \in \{1, \ldots,\ncomposite\}$, the functions $g_\idcomposite$, $\phi_{\idcomposite, \alpha_\idcomposite}$ and $\psi_{\idcomposite, \beta_\idcomposite}$ satisfy the same assumptions as the functions $g$, $\phi_\alpha$, $\psi_\beta$ given in \Cref{ass:psgla-spa,ass:gen}.
Then, for every $\idcomposite \in \{1, \ldots, \ncomposite\}$, there exists a permutation $\varrho_\idcomposite \colon \GGG_\idcomposite \to \GGG_\idcomposite$ such that, for every $\xb \in \HHH$,
\begin{equation} 
    \Db_\idcomposite \xb
    = \left( \left( D_{\idcomposite,m,n} \right)_{1 \le n \le N} \xb \right)_{1 \le m \le M_\idcomposite}
    = \varrho_\idcomposite \left(\left( \begin{matrix}
    \left(  \Db_{\idcomposite,m,k}\right)_{m \in \EE_{\idcomposite,k}}  \xb_{k}     \\[0.1cm]
    \left(  \Db_{\idcomposite, \overline{\Wc}_{\idcomposite,m}}  \xb_{\overline{\Wc}_{\idcomposite,m}} \right)_{m \in \eE_{\idcomposite, \Rc_{k}}} 
    \end{matrix} \right)_{1 \le k \le K} \right),   \label{eq:Li-dist-simple}
\end{equation}%
and, for every $\ub_\idcomposite \in \GGG_\idcomposite$, we have
\begin{align*}
    g_\idcomposite(\ub_\idcomposite) 
    &=  \sum_{k=1}^K \Big( \sum_{m \in \overline{\eE}_{\idcomposite,k}} g_{\idcomposite,m} ( u_{\idcomposite,m} ) \Big), \\
    \phi_{\idcomposite,\alpha_\idcomposite}(\Db_\idcomposite \xb, \ub_\idcomposite)
    &=  \sum_{k=1}^K \Big( 
        \sum_{m \in \eE_{\idcomposite,k}} \phi_{\idcomposite,m,\alpha_\idcomposite} \left( \Db_{\idcomposite,m,k} \xb_{\idcomposite,k}, u_{\idcomposite,m} \right) 
        + \sum_{m \in \eE_{\idcomposite, \Rc_{k}}}  \phi_{\idcomposite,m,\alpha_\idcomposite} \left( \Db_{\eV_{\idcomposite,m}} \xb_{\eV_{\idcomposite,m}}, u_{\idcomposite,m} \right) 
        \Big), \\
    \psi_{\idcomposite,\beta_\idcomposite}(\ub)
    &=  \sum_{k=1}^K \Big( 
        \sum_{m \in \overline{\eE}_{\idcomposite,k}} \psi_{\idcomposite,m,\beta_\idcomposite} \left( u_{\idcomposite,m} \right) 
        \Big).
\end{align*}{}%
Using this notation, \Cref{prop:synch-SPA-simple} can be generalized to a multi-term setting as follows.

\begin{proposition} \label{prop:synch-SPA-multi}
Assume that, for every $\idcomposite \in \{1, \ldots, \ncomposite\}$, each operator $\Db_\idcomposite$ is split over workers $\{1, \ldots, K\}$ such that, for every $k\in \{1, \ldots, K\}$, $(\Db_{\idcomposite,m,k})_{m \in \overline{\eE}_{\idcomposite,k}}$ is stored on worker $k$. 
Let, for every $k\in \{1, \ldots, K\}$, $\xb_k^{(0)} \in \HHH_k$, $ \zb_{\idcomposite,k}^{(0)} \in \GGG_{\idcomposite,k}$, and $ \ub_{\idcomposite,k}^{(0)} \in \GGG_{\idcomposite,k}$. 
Let $(\xb^{(t)})_{1 \le t \le T}$ and $(\zb_\idcomposite^{(t)}, \ub_\idcomposite^{(t)})_{1 \le \idcomposite \le \ncomposite, 1 \le t \le T}$ be samples generated by 
Algorithm~\ref{algo:multi-synch},
where $\gamma \in \big]0, \big(\lambda_h + \sum_{\idcomposite=1}^\ncomposite (\| \Db_\idcomposite \|^2/\alpha_\idcomposite^2)\big)^{-1} \big[$, for every $k\in \{1, \ldots, K\}$, $(\wb^{(t)}_k)_{1 \le t \le T}$ is a sequence of i.i.d. standard Gaussian random variables in $\HHH_k$. 
In addition, 
for every $\idcomposite \in \{1, \ldots, \ncomposite\}$, 
\begin{align*}
       \distc{\idcomposite,k, \alpha_\idcomposite}{\zb_{\idcomposite,k}}{\vb_{\idcomposite,k}, \ub_{\idcomposite,k}}  
    &   \propto \exp \Bigg( - \sum_{m\in \eE_{\idcomposite,k}} \left( g_{\idcomposite,m}(z_{\idcomposite,m}) + \phi_{\idcomposite,m,\alpha_\idcomposite}( \Db_{\idcomposite,m,k} \xb_k  , z_{\idcomposite,m} - u_{\idcomposite,m} ) \right)    \\
    &   \qquad  
        - \sum_{m \in \eE_{\idcomposite,\Rc_k}} \left( g_{\idcomposite,m}(z_{\idcomposite,m}) + \phi_{\idcomposite,m,\alpha_\idcomposite}( \Db_{\idcomposite,m,\overline{\Wc}_{\idcomposite,m}} \xb_{\overline{\Wc}_{\idcomposite,m}}  , z_{\idcomposite,m} - u_{\idcomposite,m} ) \right) \Bigg),
\end{align*}
where, for every $\xb \in \HHH$, $\vb_{\idcomposite,k} = (v_{\idcomposite,m})_{m \in \overline{\eE}_{\idcomposite,k}} = \varrho_c 
\left(\begin{matrix}
( \Db_{\idcomposite,m,k} \xb_k )_{m \in \eE_{\idcomposite,k}}  \\
( \Db_{\idcomposite,m,\overline{\Wc}_{\idcomposite,m}} \xb_{\overline{\Wc}_{\idcomposite,m}} )_{m \in \eE_{\idcomposite,\Rc_k}}
\end{matrix}\right)$. \\
Then Algorithm~\ref{algo:multi-synch} is equivalent to Algorithm~\ref{algo:SPA-multi-sampling}.
\end{proposition}

\begin{algorithm}[htbp]
    \small

    \For{$k = 1$ \KwTo $K$}{
        \For{$k' \in \Sc_k$}{
            Send $(x_n^{(0)})_{n \in \cup_{\idcomposite =1}^\ncomposite \eV_{\idcomposite,(k',k)}}$ to worker $k'$\;
        }

        \For{$k' \in \Rc_k$}{
            Receive  $(x_n^{(0)})_{n \in \cup_{\idcomposite =1}^\ncomposite \eV_{\idcomposite,(k,k')}}$ from worker $k'$\;
        }

        \For{$\idcomposite = 1$ \KwTo $\ncomposite$}{
            $\vb_{\idcomposite,k}^{(0)} = \varrho_\idcomposite
                \left(\begin{matrix}
                ( \Db_{\idcomposite,m,k} \xb_k^{(0)} )_{m \in \eE_{\idcomposite,k}}  \\
                ( \Db_{\idcomposite,m,\overline{\Wc}_{\idcomposite,m}} \xb_{\overline{\Wc}_{\idcomposite,m}}^{(0)} )_{m \in \eE_{\idcomposite,\Rc_k}}
                \end{matrix}\right)$\;
        }
    }

    \For{$t = 0$ \KwTo $T$}{
        \For{$k = 1$ \KwTo $K$}{

            \For{$\idcomposite = 1$ \KwTo $\ncomposite$}{
                $(d_{\idcomposite,m}^{(t)})_{m \in \overline{\eE}_{\idcomposite,k}} 
                    = \left( \phi_{\idcomposite, m,\alpha_i}' (\cdot, z_{\idcomposite,m}^{(t)}-u_{\idcomposite,m}^{(t)}) (v_{\idcomposite,m}^{(t)}) \right)_{m \in \overline{\eE}_{\idcomposite,k}}$;
            }

            \For{$k' \in \Rc_k$}{
                $\displaystyle \widetilde{\db}_{(k',k)}^{(t)} = \sum_{\idcomposite =1}^\ncomposite \sum_{m \in \eE_{\idcomposite, (k',k)}} \Db_{\idcomposite,m,k'}^* d_{\idcomposite,m}^{(t)}$; \\
                Send $\widetilde{\db}_{(k',k)}^{(t)}$ to worker $k'$;
            }

            \For{$k' \in \Sc_k$}{
                Receive $\widetilde{\db}_{(k,k')}^{(t)}$ from worker $k'$;
            }

            $\displaystyle \deltab_k^{(t)} = \sum_{\idcomposite =1}^\ncomposite \sum_{m \in \overline{\eE}_{\idcomposite,k}} \Db_{\idcomposite,m,k}^* d_{\idcomposite,m}^{(t)} + \sum_{k'\in \Sc_k} \widetilde{\db}_{(k,k')}^{(t)}$; \\
            $\displaystyle \xb_k^{(t+1)} = \prox_{\gamma f_k}\left( \xb_k^{(t)} - \gamma \nabla h_k(\xb_k^{(t)}) - \gamma \deltab_k^{(t)} + \sqrt{2\gamma} \wb_k^{(t)} \right)$; \\

            \For{$k' \in \Sc_k$}{
                Send $(x_n^{(t+1)})_{n \in \cup_{\idcomposite =1}^\ncomposite \eV_{\idcomposite, (k',k)}}$ to worker $k'$;
            }

            \For{$k' \in \Rc_k$}{
                Receive $(x_n^{(t+1)})_{n \in \cup_{\idcomposite =1}^\ncomposite \eV_{\idcomposite, (k,k')}}$ from worker $k'$;
            }

            \For{$\idcomposite = 1$ \KwTo $\ncomposite$}{
                $\displaystyle \vb_{\idcomposite,k}^{(t+1)} = 
                \varrho_\idcomposite \left(\begin{matrix}
                ( \Db_{\idcomposite,m,k} \xb_k^{(t+1)} )_{m \in \eE_{\idcomposite,k}}  \\
                ( \Db_{\idcomposite,m,\overline{\Wc}_{\idcomposite,m}} \xb_{\overline{\Wc}_{\idcomposite,m}}^{(t+1)} )_{m \in \eE_{\idcomposite,\Rc_k}}
                \end{matrix}\right)$; \\
                $\zb_{\idcomposite,k}^{(t+1)} \sim 
                \distc{\idcomposite, \alpha_\idcomposite, k}{\zb_{\idcomposite,k}}{\vb_{\idcomposite,k}^{(t+1)}, \ub_{\idcomposite,k}^{(t)}}$; \\
                $\ub_{\idcomposite,k}^{(t+1)} \sim \Nc \Big( \frac{\beta_i^2}{\alpha_i^2+\beta_i^2} (\zb_{\idcomposite,k}^{(t+1)} - \vb_{\idcomposite,k}^{(t+1)}), \frac{\alpha_i^2+\beta_i^2}{\alpha_i^2 \beta_i^2} \Idb \Big)$;
            }
        }
    }
    \caption{\small Proposed distributed SGS with multiple composite terms $\ncomposite > 1$.}
    \label{algo:multi-synch}
\end{algorithm}

\begin{proof}
    The proof of Proposition~\ref{prop:synch-SPA-multi} is similar to the one of Proposition~\ref{prop:synch-SPA-simple}.
\end{proof}

} \fi

\bibliographystyle{siamplain}
\bibliography{abbr,references}

\end{document}